\documentclass[12pt, draftclsnofoot, onecolumn]{IEEEtran} 
\usepackage{amsmath,amssymb}
\usepackage{subfigure}
\usepackage{graphicx,graphics,color,psfrag}
\usepackage{cite,balance}
\usepackage{algorithm}
\usepackage{accents}
\usepackage{amsthm}
\usepackage{bm}
\usepackage{url}
\usepackage{algorithmic}
\usepackage[english]{babel}
\usepackage{multirow}
\usepackage{enumerate}
\usepackage{cases}
\usepackage{stfloats}
\usepackage{dsfont}
\usepackage{color,soul}
\usepackage{amsfonts}
\usepackage{tcolorbox}

\usepackage{lipsum}
\usepackage{mathtools}
\usepackage{cuted}

\usepackage{cite,graphicx,amsmath,amssymb}
\usepackage{fancyhdr}
\usepackage{hhline}
\usepackage{graphicx,graphics}
\usepackage{array,color}
\usepackage{amsmath}
\usepackage{amsthm}
\usepackage[flushleft]{threeparttable}
\IEEEoverridecommandlockouts

\newtheorem{proposition}{Proposition}
\newtheorem{remark}{Remark}

\newtheorem{lemma}{Lemma}

\include{header}

\newcommand{\mv}[1]{\mbox{\boldmath{$ #1 $}}}

\makeatletter
\def\endthebibliography{%
	\def\@noitemerr{\@latex@warning{Empty `thebibliography' environment}}%
	\endlist
}
\makeatother

\begin{document}
	
	\title{MIMO Integrated Sensing and Communication: CRB-Rate Tradeoff}
	\author{Haocheng Hua, Tony Xiao Han, and Jie Xu \\
		\thanks{Part of this paper will be presented in IEEE Global Communications Conference (GLOBECOM), Rio de Janeiro, Brazil, December 4-8, 2022 \cite{hua2022mimo}.} 	
		\thanks{H. Hua and J. Xu are with the School of Science and Engineering (SSE) and the Future Network of Intelligence Institute (FNii), The Chinese University of Hong Kong (Shenzhen), Shenzhen, China (e-mail: haochenghua@link.cuhk.edu.cn, xujie@cuhk.edu.cn).  J.~Xu is the corresponding author.}
		\thanks{T. X. Han is with the 2012 lab, Huawei, Shenzhen 518129, China (e-mail: tony.hanxiao@huawei.com). }
	}

	\vspace{-1cm}
	\markboth{}{}
	\maketitle
	
	\setlength\abovedisplayskip{2pt}
	\setlength\belowdisplayskip{2pt}

	%
	
	\begin{abstract}
		This paper studies a multiple-input multiple-output (MIMO) integrated sensing and communication (ISAC) system, in which a multi-antenna base station (BS) sends unified wireless signals to estimate one sensing target and communicate with a multi-antenna communication user (CU) simultaneously. We consider two sensing target models, namely the point and extended targets, respectively. For the point target case, the BS estimates the target angle and the reflection coefficient as unknown parameters, and we adopt the Cram\'er-Rao bound (CRB) for angle estimation as the sensing performance metric. For the extended target case, the BS estimates the complete target response matrix, and we consider three different sensing performance metrics including the trace, the maximum eigenvalue, and the determinant of the CRB matrix for target response matrix estimation. For each of the four scenarios with different CRB measures, we investigate the fundamental tradeoff between the estimation CRB for sensing and the data rate for communication, by characterizing the Pareto boundary of the achievable CRB-rate (C-R) region. In particular, we formulate a new MIMO rate maximization problem for each scenario, by optimizing the transmit covariance matrix at the BS, subject to a different form of maximum CRB constraint and its maximum transmit power constraint. For these problems, we obtain the optimal transmit covariance solutions in semi-closed forms by using advanced convex optimization techniques.
		For the point target case, the optimal solution is obtained by diagonalizing a \emph{composite channel matrix} via singular value decomposition (SVD) together with water-filling-like power allocation over these decomposed subchannels. 
		For the three scenarios in the extended target case, the optimal solutions are obtained by diagonalizing the \emph{communication channel} via SVD, together with proper power allocation over two orthogonal sets of subchannels, one for
		both communication and sensing, and the other for dedicated sensing only. Finally, numerical results show the C-R region achieved by the optimal design in each scenario, which significantly outperforms that by other benchmark schemes such as time switching.
	\end{abstract}
	
	
	\section{Introduction}
	\label{sec:intro}
	Recently, integrated sensing and communication (ISAC) has been recognized as a candidate key technology towards sixth-generation (6G) cellular networks to enable environment-aware intelligent applications such as intelligent transportation and smart home (see, e.g., \cite{liu2022integrated,zhang2021enabling} and the references therein), in which radio signals and cellular infrastructures are reused for both sensing and communication. Motivated by the success of multiple-input multiple-output (MIMO) techniques in wireless communications \cite{telatar1999capacity,Tse2005book,heath2018foundations} and radar sensing \cite{li2007mimo,stoica2007probing,haimovich2007mimo}, MIMO ISAC has recently attracted growing research interests, in which multiple antennas can be exploited to provide the spatial multiplexing and diversity gains to increase the communication data rate and reliability \cite{telatar1999capacity,Tse2005book,heath2018foundations}, as well as the waveform and spatial diversity gains to enhance the sensing accuracy and resolution \cite{li2007mimo,stoica2007probing,haimovich2007mimo}. 
	
	In MIMO ISAC systems, the scarce spectrum and power resources are shared between the two functions of sensing and communication. As a result, there exists a fundamental tradeoff in designing the transmit strategies and resource allocation to balance the sensing and communication performances. The transmit strategies design, however, is a challenging task, due to the following two reasons in general. First, different from communication that commonly adopts the data rate as the performance metric, the sensing performance metric may vary considerably depending on specific sensing tasks (e.g., target detection or estimation) and specific parameters to be estimated (for estimation tasks of our interest). Second, the transmit design principle for MIMO communication significantly differs from that for MIMO radar, thus making it difficult to optimize the transmit strategies for balancing the sensing and communication objectives.
	
	In the literature, there have been various prior works investigating the multi-antenna ISAC by employing different performance metrics for sensing design.
	On one hand, prior works \cite{liu2018toward,liu2018mu, Eldar2020joint,xu2021rate,hua2021optimal, song2022joint, yin2022rate, wang2022noma} adopted the transmit beampattern as the sensing performance metric, based on which the transmit signal beams are focused towards the target directions to facilitate the estimation or detection.
	To be specific, the authors in \cite{liu2018toward} and \cite{liu2018mu} studied a multi-user multiple-input-single-output (MISO) ISAC system, in which the information signal beams are reused for both sensing and communication, and the transmit information beamformers are optimized to match a sensing-oriented beampattern and guarantee the communication performance simultaneously. Furthermore, to fully exploit the degrees of freedom (DoFs) provided by MIMO radar, \cite{Eldar2020joint, hua2021optimal, song2022joint} proposed to transmit the dedicated radar sensing signal beams in addition to the information beams, in which the transmit information and sensing beamformers are jointly  optimized. Moreover, advanced multiple access techniques, such as rate-splitting multiple access (RSMA) \cite{xu2021rate,yin2022rate} and non-orthogonal multiple access (NOMA) \cite{wang2022noma}, were further exploited to enhance the communication and sensing performance, in terms of data rate and sensing beampattern, respectively. 
	On the other hand, the Cram\'er-Rao bound (CRB) is another widely adopted sensing performance measure for estimation tasks (e.g., \cite{liu2021cramer,yin2022rate,song2022intelligent,li2008waveform}), which characterizes the variance lower bound by any unbiased estimators \cite{levy2008principles}. For instance, the work \cite{liu2021cramer} studied the multi-antenna ISAC system with multiple communication users (CUs) and one point or one extended target, in which the transmit information beamformers and dedicated sensing beamformers were jointly optimized to minimize the target estimation CRB, subject to a set of individual signal-to-noise-plus-interference ratio (SINR) constraints at CUs. As compared to the sensing beampattern, the estimation CRB can explicitly characterize the fundamental estimation error limit for target estimation tasks, which is thus considered in this paper. 
	
	This paper investigates the MIMO ISAC system for target estimation. We focus on characterizing the fundamental tradeoff limits between sensing and communication performances from the estimation theory and the information theory perspectives, in which the estimation CRB and the data rate are employed as the sensing and communication performance measures, respectively. This problem, however, has not been well addressed in the literature yet, even for the basic point-to-point MIMO ISAC system. To our best knowledge, only one recent work \cite{xiong2022flowing} studied the so-called CRB-rate (C-R) region for a point-to-point MIMO ISAC system with one sensing target, which is defined as the set containing all C-R pairs that can be simultaneously achieved by sensing and communication. In particular, \cite{xiong2022flowing} considered a generic target estimation model with finite sensing duration, based on which the optimal sample covariance matrix at the base station (BS) was derived for CRB minimization and rate maximization, respectively. Accordingly, the obtained solutions in \cite{xiong2022flowing} only characterized two corner points on the boundary of the C-R region, but did not characterize the whole region boundary, especially the boundary points between the two corners. This thus motivates the current work.
	
	In particular, we consider a point-to-point MIMO ISAC system, in which a multi-antenna BS sends unified wireless signals to sense one target based on the echo signal and communicate with a multi-antenna CU simultaneously. We consider two target models, namely the point and extended targets. For the point target case, the BS aims to estimate the target angle and reflection coefficient as unknown parameters, and we use the CRB for estimating the target angle as the sensing performance metric. For the extended target case, the BS aims to estimate the complete target response matrix, and we consider the trace, the maximum eigenvalue, and the determinant of the CRB matrix for estimating the target response matrix, namely the Trace-CRB, MaxEig-CRB, and Det-CRB, respectively, as three different sensing CRB metrics. For each of the four scenarios with different CRB measures for the two target models,
	we aim to reveal the \textit{complete boundary} of the C-R region, by optimizing the transmit covariance matrix at the BS. The main results of this paper are listed as follows. 
	\begin{itemize}
		\item In order to characterize the Pareto boundary of the C-R region in each scenario, we first pinpoint two corner points on each boundary by maximizing the data rate for communication only
		and minimizing the estimation CRB for sensing only, respectively. Next, to find the boundary points between the two corners on each C-R region, we formulate a new MIMO rate maximization problem by optimizing the transmit covariance matrix at the BS, subject to the corresponding maximum CRB constraint and the maximum transmit power constraint at the BS. 
		\item First, we consider the angle-CRB-constrained rate maximization problem for the point target case. We first transform this problem into a convex form and accordingly derive its optimal transmit covariance solution in a semi-closed form by using the Lagrange duality method. It is shown that the optimal solution generally follows the eigenmode transmission structure based on a \textit{composite channel matrix} composed of both the communication and sensing channels, in which the singular value decomposition (SVD) is implemented to diagonalize the composite channel, followed by the water-filling-like power allocation over the decomposed subchannels.
		\item Next, we consider the Trace-CRB, MaxEig-CRB, and Det-CRB constrained rate maximization problems for the extended target case, which are all shown to be convex. By applying advanced optimization techniques, we derive their optimal transmit covariance solutions in semi-closed forms. For all the three problems, the optimal solutions are obtained by first implementing the SVD to diagonalize the \textit{communication channel}, and then performing proper power allocation over two orthogonal sets of decomposed subchannels, one for both communication and sensing and the other (if any) for dedicated sensing only. It is shown that the optimal power allocation is monotonically increasing with respect to the equivalent channel gain of the decomposed subchannel.
		\item Finally, we present numerical results to evaluate the C-R-region boundary achieved by the optimal transmit covariance solutions under different CRB metrics for both the point and extended target cases. In the point target case, our proposed optimal design is shown to significantly outperform the benchmark scheme based on time switching between the rate-maximization and the CRB minimization designs. In the extended target case, our proposed optimal designs are shown to be superior to various benchmarks based on time switching and power splitting designs. 
		%
		%
		%
	\end{itemize}
	%
	
	
	{\it Notations:} Boldface letters refer to vectors (lower case) or matrices (upper case). For a square matrix $\mv{M}$, $\mv{M}\succeq \mv{0}$ and $\mv{M}\succ \mv{0}$ mean that $\mv{M}$ is positive semidefinite and positive definite, respectively, while  ${\operatorname{tr}}(\mv{M})$, $\operatorname{det}(\bm{M})$, $\lambda_{\text{max}}(\bm{M})$, $\lambda_{\text{min}}(\bm{M})$, $\lambda_i(\bm{M})$, and $\bm{\lambda}(\bm{M})$ denote its trace, its determinant, maximum eigenvalue, minimum eigenvalue, $i$-th eigenvalue, and the set of its eigenvalues, respectively. For an arbitrary-size matrix $\mv{M}$, ${\text{rank}}(\bm{M})$, $\bm{M}^H$, $\bm{M}^*$, $\bm{M}^T$, and $\zeta_i(\bm{M})$  denote its rank, conjugate transpose, conjugate, transpose, and  $i$-th singular value, respectively. $\otimes$, $\circ$, and $\oplus$ denote the Kronecker product, the Hadamard product, and the direct sum, respectively.
	$\mathbb{R}^{x\times y}$ and $\mathbb{C}^{x\times y}$ denotes the spaces of real and complex matrices with dimension $x \times y$, respectively. {${\mathbb{E}}\{\cdot\}$} denotes the statistical expectation. $\|\mv{x}\|$ denotes the Euclidean norm of a complex vector $\mv{x}$. $|z|$ denote the magnitude of a complex number $z$. $j = \sqrt{-1}$ denotes the imaginary unit. For a real number $x$, $\left(x\right)^+ = \max(x,0)$. $\operatorname{diag}(x_1,...,x_n)$ denotes a diagonal matrix with diagonal elements $x_1,...,x_n$. For a matrix $\bm{M} \in \mathbb{R}^{x \times y}$ (or $\mathbb{C}^{x\times y}$), $\mathcal{R}(\bm{M})$ and $\mathcal{N}(\bm{M})$ denote the range and null space of $\bm{M}$ that are subspaces of $\mathbb{R}^{x}$ (or $\mathbb{C}^x$) and $\mathbb{R}^{y}$ (or $\mathbb{C}^{y}$), respectively.

	\section{System Model}\label{sec:system1}
	
	We consider a MIMO ISAC system, in which a BS communicates with a CU and simultaneously estimates a point target or an extended target, as shown in Fig. \ref{fig:system_model_point} or Fig. \ref{fig:system_model_extended}, respectively. The BS is equipped with a uniform linear array (ULA) with $M > 1$ transmit antennas for sending ISAC signals and $N_s$ receive antennas for receiving echo signals for target estimation. The CU is equipped with $N_c > 1$ antennas.

	\begin{figure}[htb]
		\centering
		\setlength{\abovecaptionskip}{+2mm}
		\setlength{\belowcaptionskip}{-1mm}
		\subfigure[Case with a point target]{ \label{fig:system_model_point}
			\includegraphics[width=2.6in]{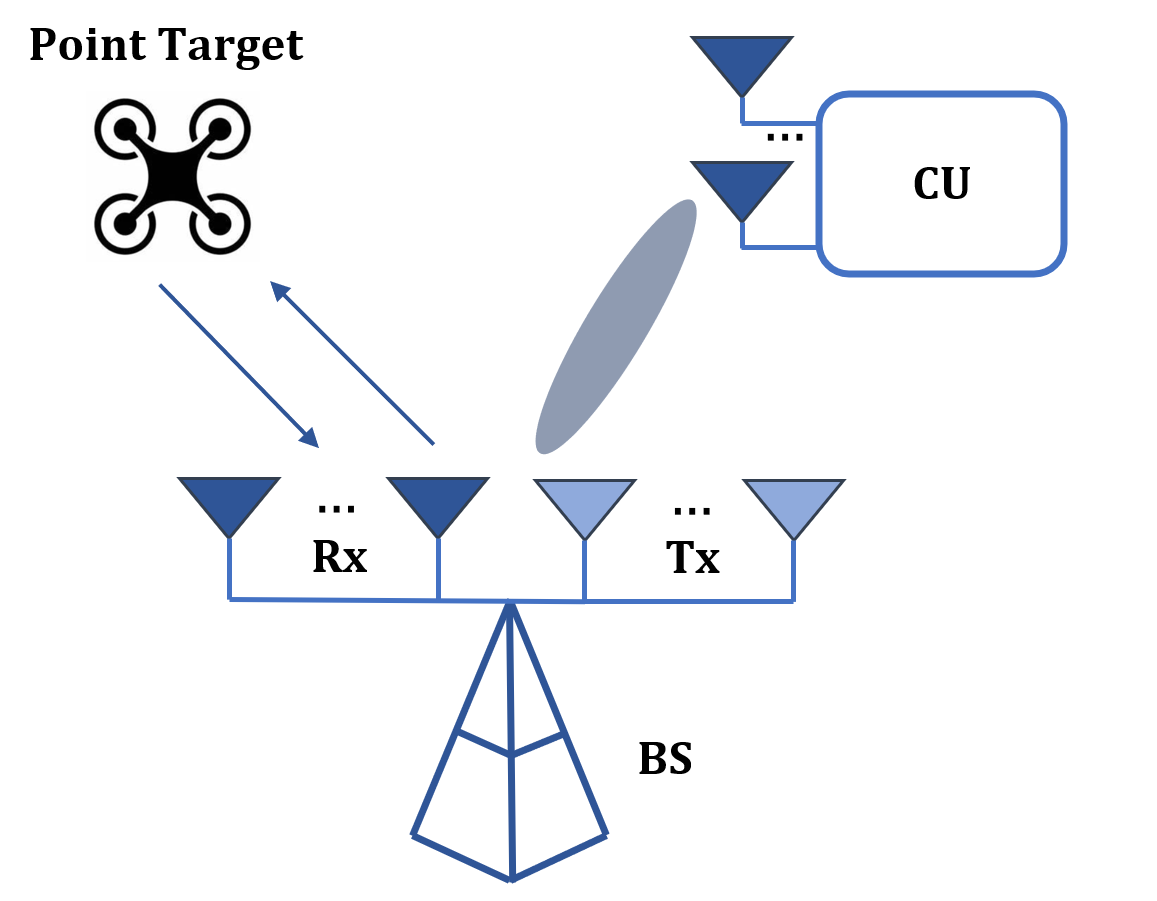}}
		\subfigure[Case with an extended target]{ \label{fig:system_model_extended}
			\includegraphics[width=2.7in]{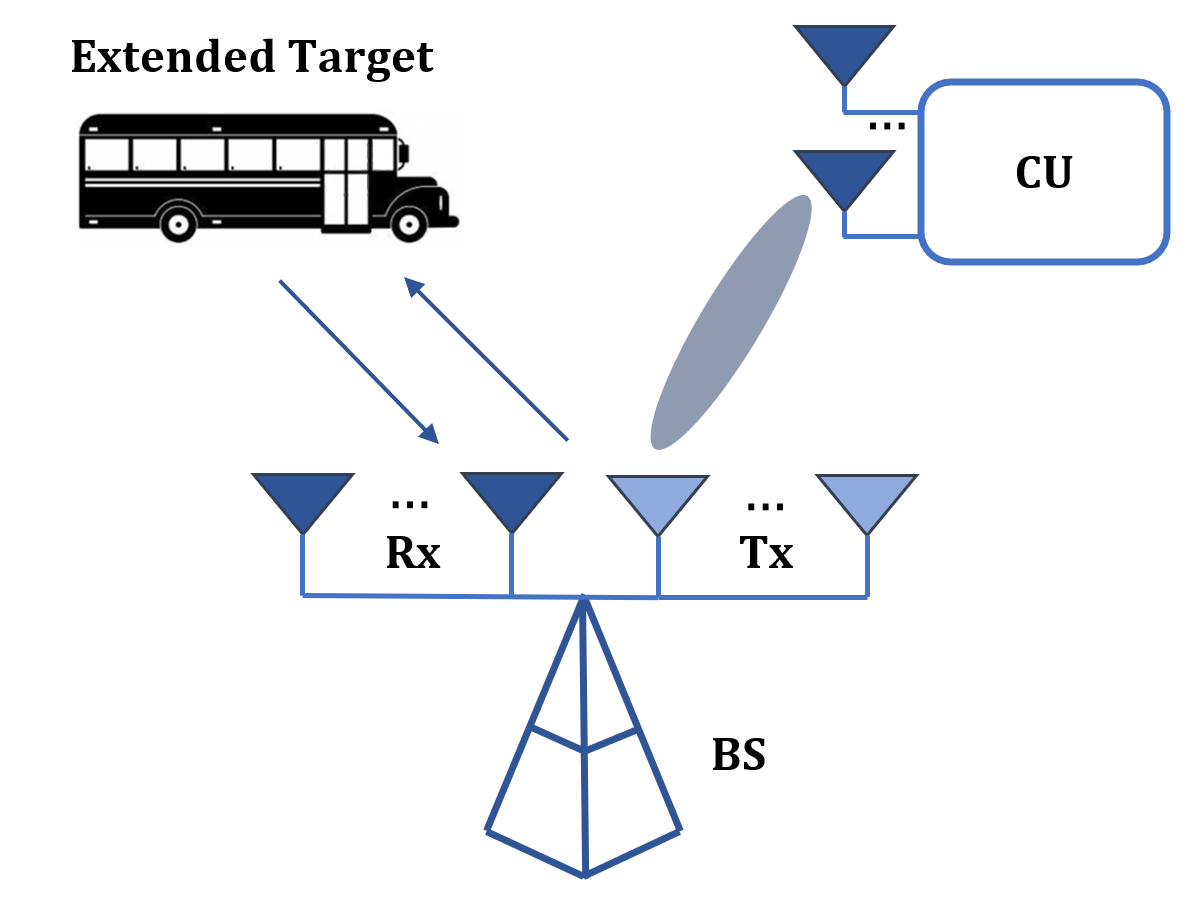}}
		\caption{Illustration of the considered MIMO ISAC system. }
		\label{fig:system_model}
	\end{figure}
	
	
	Let $\bm{x}(n)$ denote the transmit ISAC signal by the BS transmitter (BS-Tx) at symbol $n$. We consider the capacity-achieving Gaussian signaling, such that $\bm{x}(n)$ is a circularly symmetric complex Gaussian (CSCG) random vector with zero mean and covariance $\bm{Q} = \mathbb{E} \{\bm{x}(n) \bm{x}^H(n)\} \succeq \bm{0}$, i.e., $\bm{x}(n) \sim \mathcal{CN}(\bm{0}, \bm{Q})$. Let $P$ denote the transmit power budget at the BS-Tx. We thus have the transmit power constraint as
	\begin{align}\label{eqn:power:constraint}
		\operatorname{tr} (\bm{Q}) = \mathbb{E}\{\|\bm{x}(n)\|^2\} \leq P.
	\end{align}
	
	In this work, we consider a quasi-static narrowband channel model, in which the wireless channels remain unchanged over the transmission duration of our interest, as commonly adopted in the literature \cite{liu2021cramer,xu2021rate}.
	Let $\bm{H}_c \in \mathbb{C}^{N_c \times M}$ denote the channel matrix from the BS-Tx to the CU, whose rank is denoted by $r = \text{rank}(\bm{H}_c) \leq \min(N_c,M)$. Let $\bm{H}_s \in  C^{N_s \times M}$ denote the target response matrix from the BS-Tx to the target to the BS receiver (BS-Rx), which will be specified later for the cases with point and extended targets, respectively.  
	
	First, we consider the point-to-point MIMO communication from the BS-Tx to the CU. The received signal by the CU at symbol $n$ is
	\begin{align}
		\bm{y}_c(n) = \bm{H}_c \bm{x}(n) + \bm{z}_c(n),
	\end{align}
	where $\bm{z}_c(n)$ denotes the noise at the CU receiver that is a CSCG random vector with zero mean and covariance $\sigma_c^2 \bm{I}_{N_c}$, i.e., $\bm{z}_c(n) \sim \mathcal{CN}(\bm{0},\sigma_c^2 \bm{I}_{N_c})$. In this case, with Gaussian signalling, the achievable rate (in bps/Hz) of the MIMO channel with transmit covariance $\bm{Q}$ is
	\begin{align}\label{eqn:Rate}
		R(\bm{Q}) = \log_2 \det \left(\bm{I}_{N_c} + \frac{1}{\sigma_c^2} \bm{H}_c \bm{Q} \bm{H}_c^H \right).
	\end{align}
	It is assumed that the channel matrix $\bm{H}_c$ is perfectly known at the BS, such that the BS can design the transmit covariance $\bm{Q}$ based on $\bm{H}_c$ to optimize the achievable rate $R(\bm{Q})$.
	
	Next, we consider the MIMO radar sensing over a particular coherent processing interval (CPI) with $L > M$ symbols.
	Let $\mathcal{L} = \{1,\ldots, L\}$ denote the set of symbols in the CPI, and $\bm{X} = \left[\bm{x}(1),...,\bm{x}(L)\right] \in \mathbb{C}^{M \times L}$ denote the transmitted signals over the CPI.
	It is assumed that the CPI length $L$ is sufficiently long such that the sample covariance matrix $\frac{1}{L} \bm{X}\bm{X}^H$ can be approximated as the covariance matrix $\bm{Q}$, which is designed for sensing performance optimization \cite{liu2021cramer}. Accordingly, the received echo signal $\bm{Y}_s \in \mathbb{C}^{N_s \times L}$ at the BS-Rx is
	\begin{align}
		\bm{Y}_s = \bm{H}_s \bm{X} + \bm{Z}_s,
	\end{align}
	where $\bm{Z}_s \in \mathbb{C}^{N_s \times L}$ denotes the noise at the BS-Rx, with each element being an independent and identically distributed (i.i.d.) CSCG random variable with zero mean and variance $\sigma_s^2$.
	In particular, we consider the point and extended target models for $\bm{H}_s$, respectively, as detailed in the following.
	
	\subsection{Point Target Model} \label{subsection:point_target_model}
	The point target is modeled as an unstructured point that is far away from the BS \cite{liu2021cramer}. The corresponding target response matrix is
	\begin{align}\label{equ:point_target_model}
		\bm{H}_s=\alpha \mathbf{b}(\theta) \mathbf{a}^{T}(\theta) \triangleq \alpha \mathbf{A}(\theta),
	\end{align}
	where $\alpha \in \mathbb{C}$ represents the reflection coefficient that depends on both the round-trip path-loss and the radar cross section (RCS) of the target, $\theta$ is the angle of arrival (AoA)/angle of departure (AoD) of the target relative to the BS, and $\mathbf{a}(\theta) \in \mathbb{C}^{M \times 1}$ and $\mathbf{b}(\theta) \in \mathbb{C}^{N_{s} \times 1}$ denote the steering vectors of the transmit and receive antennas, respectively.
	By choosing the center of the ULA as the reference point and assuming half-wavelength spacing between adjacent antennas \cite{liu2021cramer}, we have 
	\begin{align}
		\label{equ:Tx_steer}
		\mathbf{a}(\theta) &= \left[e^{-j \frac{M-1}{2} \pi \sin \theta}, e^{-j \frac{M-3}{2} \pi \sin \theta},...,e^{j \frac{M-1}{2} \pi \sin \theta}\right]^T, \\
		\label{equ:Rx_steer}
		\mathbf{b}(\theta) &= \left[e^{-j \frac{N_s-1}{2} \pi \sin \theta}, e^{-j \frac{N_s-3}{2} \pi \sin \theta},...,e^{j \frac{N_s-1}{2} \pi \sin \theta}\right]^T.
	\end{align}
	For the point target case, the BS needs to estimate the complex coefficient $\alpha$ and the angle $\theta$ as unknown parameters. As it is difficult to extract the target information from $\alpha$, we focus on the estimaton of $\theta$ \cite{song2022intelligent}. In this case, the CRB for estimating $\theta$ is expressed as \cite{liu2021cramer}
	\begin{align} \label{equ:point_target_CRLB}
		\text{CRB}_1(\bm{Q}) =  
		\frac{\sigma_{s}^{2} \operatorname{tr}\left(\mathbf{A}^{H}(\theta) \mathbf{A}(\theta) \bm{Q}\right)}{2|\alpha|^{2} L\left(\operatorname{tr}\left(\dot{\mathbf{A}}^{H}(\theta) \dot{\mathbf{A}}(\theta) \bm{Q}\right) \operatorname{tr}\left(\mathbf{A}^{H}(\theta) \mathbf{A}(\theta) \bm{Q}\right)-\left|\operatorname{tr}\left(\dot{\mathbf{A}}^{H}(\theta) \mathbf{A}(\theta) \bm{Q}\right)\right|^{2}\right)},
	\end{align}
	where $\dot{\mathbf{A}}(\theta) =  \frac{\partial {\mathbf{A}}(\theta)}{\partial \theta} = \mathbf{b}(\theta) \mathbf{\dot{a}}^T(\theta) + \dot{\mathbf{b}}(\theta) \mathbf{a}^T(\theta)$. Here,
	$\dot{\mathbf{a}}(\theta)$ and $\dot{\mathbf{b}}(\theta)$ denote the derivatives of $\mathbf{a}(\theta)$ and $\mathbf{b}(\theta)$, respectively, i.e.,
	\begin{align}
		\dot{\mathbf{a}}(\theta) & = \left[-j a_1 \frac{M-1}{2} \pi \cos \theta,...,j a_M \frac{M-1}{2} \pi \cos \theta \right]^T, \\
		\dot{\mathbf{b}}(\theta) & = \left[-j b_1 \frac{N_s-1}{2} \pi \cos \theta,...,j b_{N_s} \frac{N_s-1}{2} \pi \cos \theta \right]^T,
	\end{align}
	where $a_i$ and $b_i$ are the $i$th entry of $\mathbf{a}(\theta)$ and that of $\mathbf{b}(\theta)$, respectively. Notice that by the symmetry of the ULA, $\mathbf{a}(\theta)$ is orthogonal to $\dot{\mathbf{a}}(\theta)$ and $\mathbf{b}(\theta)$ is orthogonal to $\dot{\mathbf{b}}(\theta)$ regardless of $\theta$, i.e., $\mathbf{a}^H(\theta) \dot{\mathbf{a}}(\theta) = 0$ and $ \mathbf{b}^H(\theta) \dot{\mathbf{b}}(\theta) = 0, \forall \theta$. 
	
	\subsection{Extended Target Model} \label{subsection:extended_target_model}
	In general, the extended target is modeled as the combination of a large number of $K$ distributed point-like scatterers. In this case, $\bm{H}_s$ is expressed as \cite{liu2021cramer}
	\begin{align}\label{equ:extende_target_model}
		\bm{H}_s=\sum_{k=1}^{K} \alpha_{k} \bm{b}\left(\theta_{k}\right) \bm{a}^{T}\left(\theta_{k}\right),
	\end{align}
	where $\alpha_{k}$ denotes the reflection coefficient of the $k$-th scatterer, $\theta_{k}$ denotes its associated AoA/AoD relative to the BS, and $\bm{a}\left(\theta_{k}\right)$ and $\bm{b}\left(\theta_{k}\right)$ denote the corresponding transmit and receive steering vectors given in (\ref{equ:Tx_steer}) and (\ref{equ:Rx_steer}), respectively. As the number of scatterers $K$ may not be available \textit{a priori}, the objective of sensing is to estimate the complete target response matrix $\bm{H}_s$ with $M N_s$ complex parameters, based on which the BS may extract the parameters of each scatterer using algorithms such as the multiple signal classification (MUSIC) \cite{schmidt1986multiple} and the amplitude and phase estimation (APES)      \cite{li1996adaptive}. 
	In this case, the CRB matrix for estimating $\bm{H}_s$ is given by \cite{liu2021cramer}
	\begin{align}\label{eqn:CRB_matrix}
		\overline{\bold{CRB}}(\bm{Q}) = \bm{J}(\bm{Q})^{-1},
	\end{align}
	where $\bm{J}(\bm{Q})$ is the Fisher information matrix given by \cite{liu2021cramer}
	\begin{align}
		\bm{J}(\bm{Q})=
		\frac{L}{\sigma_{s}^{2}} \bm{Q}^T \otimes \bm{I}_{N_{s}}.
	\end{align}
	Note that $\overline{\bold{CRB}}(\bm{Q})$ in (\ref{eqn:CRB_matrix}) is a complex matrix with dimension $M N_s \times M N_s$, with the $(M(i-1)+j)$-th diagonal element representing the lower bound of variance for unbiasedly estimating the $(i,j)$-th element of $\bm{H}_s$, $1 \leq i \leq N_s, 1 \leq j \leq M$.
	To facilitate the ISAC system design, we adopt three different types of scalar CRB metrics based on the trace, maximum eigenvalue, and determinant of the CRB matrix $\overline{\bold{CRB}}(\bm{Q})$ \cite{li2008waveform}, namely Trace-CRB, MaxEig-CRB, and Det-CRB, which are defined in (\ref{equ:Trace_opt}), (\ref{equ:Eigen_opt}), and (\ref{equ:Det_opt}), respectively.
	\begin{align}\label{equ:Trace_opt}
		&\text{CRB}_2(\bm{Q}) = \operatorname{tr}(\overline{\bold{CRB}}(\bm{Q})) = \frac{\sigma_{s}^2 N_s}{L} \operatorname{tr}(\bm{Q}^{-1})\\
		\label{equ:Eigen_opt}
		&\text{CRB}_3(\bm{Q}) = \lambda_{\text{max}} (\overline{\bold{CRB}}(\bm{Q})) = \frac{\sigma_s^2}{L} 		\lambda_{\text{max}} (\bm{Q}^{-1})\\
		\label{equ:Det_opt}
		&\text{CRB}_4(\bm{Q}) = \operatorname{det}(\overline{\bold{CRB}}(\bm{Q})) = (\frac{\sigma_s^2}{L})^{M N_s} \operatorname{det}(\bm{Q}^{-1})^{N_s}
	\end{align}
	Intuitively, minimizing the Trace-CRB (i.e., $\text{CRB}_2(\bm{Q})$ in (\ref{equ:Trace_opt})) corresponds to minimizing the sum CRB for estimating the elements of $\bm{H}_s$. Next, 
	minimizing the MaxEig-CRB (i.e., $\text{CRB}_3(\bm{Q})$ in (\ref{equ:Eigen_opt})) ensures the fairness for estimating different elements of $\bm{H}_s$ by minimizing the upper bound of the worst-case CRB. Furthermore, as $\operatorname{det}(\overline{\bold{CRB}}(\bm{Q})) = \prod_{i=1}^{M N_s}  \lambda_i(\overline{\bold{CRB}}(\bm{Q}))$, 
	minimizing the Det-CRB ($\text{CRB}_4(\bm{Q})$ in (\ref{equ:Det_opt})) is equivalent to minimizing $\sum_{i=1}^{M N_s} \ln \lambda_i(\overline{\bold{CRB}}(\bm{Q}))$, which ensures the proportional fairness for estimating different elements of $\bm{H}_s$ \cite{heath2018foundations}.
	
	For notational convenience, we define the point target case with $\text{CRB}_1(\bm{Q})$ in (\ref{equ:point_target_CRLB}) as Scenario 1, and the three extended target cases with $\text{CRB}_2(\bm{Q})$ in (\ref{equ:Trace_opt}), $\text{CRB}_3(\bm{Q})$ in (\ref{equ:Eigen_opt}), and $\text{CRB}_4(\bm{Q})$ in (\ref{equ:Det_opt}) as Scenarios 2, 3, and 4, respectively. We will characterize the C-R regions for the four scenarios next. 

	\section{C-R Region Characterization}\label{section:two_boundary_point} 
	
	Our objective is to reveal the fundamental tradeoff between the data rate $R(\bm{Q})$ in \eqref{eqn:Rate} for communication and the estimation CRB $\text{CRB}_i(\bm{Q})$ for sensing in Scenario $i \in \{1,...,4\}$. To start with, we define the C-R region, which is a set containing all C-R pairs that can be simultaneously achieved by the ISAC system under the given transmit power constraint in (\ref{eqn:power:constraint}). Mathematically, the C-R region with power budget $P$ in Scenario $i \in \{1,2,3,4\}$ is defined as
	\begin{align}
		\nonumber
		\mathcal{C}_i^{\text{C-R}}(P)  \triangleq \left\{ (\bar{\Gamma},\bar{R}):  \bar{\Gamma} \ge \text{CRB}_i(\bm{Q}),
		\bar{ R} \leq \log_2 \det \left(\bm{I}_{N_c} + \frac{1}{\sigma_c^2} \bm{H}_c \bm{Q} \bm{H}_c^H \right),  \operatorname{tr}(\bm{Q}) \leq P, \bm{Q} \succeq \bm{0}  \right\} .
	\end{align}
	
	We are particularly interested in finding the Pareto boundary of C-R region $\mathcal{C}_i^{\text{C-R}}(P), i \in \{1,2,3,4\}$. Towards this end, we consider the following CRB-constrained rate maximization problem (P$i$) for each scenario $i \in \{1,2,3,4\}$.
	\begin{align} 
		\label{equ:general_ex}
		\text{(P$i$)}: \text{ } \max _{\bm{Q} \succeq \bm{0}}  \text{ } \log_2 \det \left(\bm{I}_{N_c} + \frac{1}{\sigma_c^2} \bm{H}_c \bm{Q} \bm{H}_c^H \right), \quad
		\text { s.t. }  \text{CRB}_i(\bm{Q}) \leq \Gamma_i, \quad
		\operatorname{tr}(\bm{Q}) \leq P. 
	\end{align}

	To facilitate the whole boundary characterization for each of the four different scenarios 
	and gain more insights,
	in the following, we first pinpoint two corner points on the boundary of each C-R region $\mathcal{C}_i^{\text{C-R}}(P)$, which correspond  to rate maximization for communication only and CRB minimization for sensing only, respectively.
	
	\subsection{Rate-Maximization Corner Point for Communication Only}
	
	First, we find the rate-maximization corner point with communication only for each of the four scenarios. Towards this end, we maximize the achievable rate $R(\bm{Q})$ by optimizing the transmit covariance $\bm{Q}$, subject to the transmit power constraint, i.e.,
	\begin{align}\label{equ:P_R}
		\max _{\bm{Q}\succeq \bm{0}}  \text{ } \log_2 \det \left(\bm{I}_{N_c} + \frac{1}{\sigma_c^2} \bm{H}_c \bm{Q} \bm{H}_c^H \right), \quad \text { s.t. } \operatorname{tr}(\bm{Q}) \leq P.
	\end{align}
	To facilitate the derivation, we express the SVD of $\bm{H}_c$ as $\bm{H}_c = \mv{U}_c \mv{\Sigma}_c \mv{V}_c^H$, where $\mv{U}_c \in \mathbb{C}^{N_c \times N_c}$ and $\mv{V}_c \in \mathbb{C}^{M \times M}$ with $\mv{U}_c^H \mv{U}_c =\mv{U}_c \mv{U}^H_c = \bm{I}_{N_c}$ and $\mv{V}_c^H \mv{V}_c = \mv{V}_c \mv{V}_c^H = \bm{I}_{M}$, and $\mv{\Sigma}_c \in \mathbb{C}^{N_c \times M}$ is an all-zero matrix except the first $r$ diagonal elements being the $r$ non-zero singular values $\zeta_1(\bm{H}_c) \geq ... \geq \zeta_r(\bm{H}_c) > 0$. It has been well established in \cite{telatar1999capacity} that the optimal solution to problem (\ref{equ:P_R}) is given by $\bm{Q}_c^* = \mv{V}_c \bm{\Lambda}_c \mv{V}_c^H$, where $\bm{\Lambda}_c = \operatorname{diag}(p_{c,1}^*,...,p_{c,r}^*,0,...,0)$ denotes the water-filling power allocation matrix with its first $r$ diagonal elements given by
	\begin{align}\label{eq:WF}
		p_{c,k}^* = \left(\nu - \frac{\sigma_c^2}{\zeta_k^2(\bm{H}_c)}\right)^+, \forall k \in \{1,\ldots, r\}.
	\end{align}
	In (\ref{eq:WF}), $\nu$ is the water level that can be obtained based on $\sum_{k=1}^r p_{c,k}^* = P$. 
	
	At the obtained $\bm{Q}_c^*$, let  $R_{\text{max}} = R(\bm{Q}_c^*) =  \sum_{k=1}^{r} \log_2 (1+\frac{\zeta_k^2(\bm{H}_c) p_{c,k}^*}{\sigma_c^2})$ denote the maximum achievable rate and $\text{CRB}_{C,i} = \text{CRB}_i(\bm{Q}_c^*)$ denote the correspondingly achieved estimation CRB in Scenario $ i \in \{1,2,3,4\}$. As a result, we obtain the rate-maximization corner point on the boundary of C-R region $\mathcal{C}_i^{\text{C-R}}(P)$ as $(\text{CRB}_{C,i}, R_{\text{max}}), i \in \{1,2,3,4\}$.
	
	\begin{remark} \label{remark:finite_SCRB}
		\emph{Note that for Scenario 1 with point target, it follows from (\ref{equ:point_target_CRLB}) that $\text{CRB}_{C,1}$ will become undefined if $\bm{a}^*(\theta)$ is  orthogonal to the range space of rate-maximization covariance $\bm{Q}_c^*$, i.e.,  $\mathcal{R}(\bm{Q}_c^*)$. In practice, however, $\bm{H}_c$ is generally drawn following certain random distributions due to channel fading, and thus $\bm{a}^*(\theta)$ will be non-orthogonal to $\mathcal{R}(\bm{Q}_c^*)$ with probability one. In this case, $\text{CRB}_{C,1}$ is well defined and bounded, and therefore, the parameter $\theta$ is estimable. On the other hand, in the case with extended target, it follows from (\ref{equ:Trace_opt}), (\ref{equ:Eigen_opt}), and (\ref{equ:Det_opt}) that if $\bm{Q}_c^*$ is rank-deficient (i.e., $\text{rank}(\bm{Q}_c^*) < M$), then $\text{CRB}_{C,i} \rightarrow \infty, i \in \{2,3,4\}$. This means that $\bm{H}_s$ is not estimable in this case due to the lack of DoFs, and the corresponding rate-maximization corner point becomes $(\infty, R_{\text{max}})$ for Scenarios 2, 3, and 4. In general, this happens when the communication channel $\bm{H}_c$ is rank-deficient (i.e., {$\text{rank}(\bm{H}_c) = r < M$}) or the transmit power at the BS is smaller than a certain threshold (i.e., $P \le {P_0 \triangleq } \sum_{i=1}^{M-1}  (\frac{\sigma_c^2}{\zeta_M^2(\bm{H}_c)}-\frac{\sigma_c^2}{\zeta_i^2(\bm{H}_c)})$) under $\text{rank}(\bm{H}_c) = M$. }
	\end{remark}

	\subsection{CRB-Minimization Corner Point for Scenario 1 with Point Target}
	
	
	Next, we find the other CRB-minimization corner point on the Pareto boundary of each C-R region. We first consider Scenario 1 for the  point target case in this subsection, and will address Scenarios 2-4 for the extended target case in the next subsection. For the point target case, we formulate the CRB minimization problem as
	\begin{align}\label{equ:point_target_formulation_fe}
		\text{ } \min _{\bm{Q} \succeq \bm{0}} \text{ } 
		\text{CRB}_1(\bm{Q}), \quad
		\text{s.t. }  \operatorname{tr}(\bm{Q}) \leq P. 
	\end{align}
	We have the optimal solution to problem (\ref{equ:point_target_formulation_fe}) given in Proposition \ref{Pro:Q_min_CRB} in the following, in which we rewrite $\bm{a}(\theta)$ as $\bm{a}$ for notational convenience. Note that Proposition \ref{Pro:Q_min_CRB} has been proved in \cite{li2008waveform}, for which the detailed proof is omitted for brevity. 
	
	\begin{proposition}\label{Pro:Q_min_CRB}
		\emph{The optimal solution to problem (\ref{equ:point_target_formulation_fe}), denoted by $\bm{Q}_{s,1}^*$, is given as follows by considering three cases.
			\begin{itemize}
				\item When $N_s > M$ or $\|\dot{\bm{b}}\| > \|\dot{\bm{a}}\|$, we have $\bm{Q}_{s,1}^* = \frac{P}{\|\bm{a}\|_2^2} \bm{a}^* \bm{a}^T$;
				\item When $N_s = M$ or $\|\dot{\bm{b}}\| = \|\dot{\bm{a}}\|$, we have $\bm{Q}_{s,1}^* = P \eta \frac{\dot{\bm{a}}^* \dot{\bm{a}}^T}{\|\dot{\bm{a}}\|_2^2} + (1-\eta)P \frac{\bm{a}^* \bm{a}^T}{\|\bm{a}\|_2^2}$ for any $0 \leq \eta < 1$;
				\item When $N_s < M$ or  $\|\dot {\bm{b}}\| <  \|\dot {\bm{a}}\|$, we have $\bm{Q}_{s,1}^* = P \eta \frac{\dot{\bm{a}}^* \dot{\bm{a}}^T}{\|\dot{\bm{a}}\|_2^2} + (1-\eta)P \frac{\bm{a}^* \bm{a}^T}{\|\bm{a}\|_2^2}$ with $\eta \rightarrow 1$. 
		\end{itemize}}
	\end{proposition}

	\begin{remark} \label{remark:CRB_min_Q}
		\emph{ 
			According to Proposition \ref{Pro:Q_min_CRB}, the minimum CRB is obtained as $\text{CRB}_{1,\text{min}} \triangleq \text{CRB}_1(\bm{Q}_{s,1}^*)$.
			Notice that when $N_s > M$, $\bm{Q}_{s,1}^*$ is unique with $\text{rank}(\bm{Q}_{s,1}^*) = 1$, and the correspondingly achieved data rate is given as $R_{1,S} = R(\bm{Q}_{s,1}^*)$. When $N_s = M$, $\bm{Q}_{s,1}^*$ is not unique and we can optimize $0 \leq \eta < 1$ to obtain $R_{1,S} = \max_{0 \leq \eta <1} R(\bm{Q}_{s,1}^*)$. When $N_s < M$, however, $R_{1,S}$ can only be asymptotically achieved, i.e., $R_{1,S} = \lim_{\eta \rightarrow 1} R(\bm{Q}_{s,1}^*)$. Combining the above, the CRB-minimization corner point on the Pareto boundary of $\mathcal{C}_1^{\text{C-R}}(P)$ is obtained as $(\text{CRB}_{1,\text{min}}, R_{1,S})$.}
	\end{remark}
	
	\subsection{CRB-Minimization Corner Point for Scenarios 2-4 with Extended Target}
	
	In this subsection, we obtain the CRB-minimization corner point on the boundary of each C-R region $\mathcal{C}_i^{\text{C-R}}(P), i \in \{2,3,4\}$, with the extended target. Towards this end, we have the CRB minimization problem for Scenario $i \in \{2,3,4\}$ as 
	\begin{align}\label{equ:P_CRB_combined}
		\min _{\bm{Q}\succeq \bm{0}}  \text{ } \text{CRB}_i(\bm{Q}), \quad
		\text {s.t. }  \operatorname{tr}(\bm{Q}) \leq P. 
	\end{align}
	We then have the following proposition.
	\begin{proposition} \label{Pro:three_CRB_min}
		\emph{The optimal solutions to the three problems in (\ref{equ:P_CRB_combined}) are identical, given by $\bm{Q}_{s,i}^{*} = \frac{P}{M} \bm{I}_M, \forall i \in \{2,3,4\}$.}
		\begin{proof}
			See Appendix \ref{three_CRB_min_proof}.
		\end{proof}
	\end{proposition}
	
	Based on Proposition \ref{Pro:three_CRB_min}, the achieved Trace-CRB, MaxEig-CRB, and Det-CRB are given by $\text{CRB}_{2,\text{min}} = \frac{\sigma_{s}^2 N_s M^2}{P L}$, $\text{CRB}_{3,\text{min}} = \frac{M \sigma_s^2}{L P}$, and $\text{CRB}_{4,\text{min}} = (\frac{M \sigma_s^2}{L P})^{M N_s}$ for Scenario 2, 3, and 4, respectively. The corresponding communication data rate are identical for the three scenarios, i.e., $R_{i,S} = \sum_{k=1}^r \log_2 \left( 1+ \frac{\zeta_k^2(\bm{H}_c) P}{\sigma_c^2 M} \right), i \in \{2,3,4\}$. As a result, the CRB-minimization corner points for C-R region $\mathcal{C}_2^{\text{C-R}}(P)$ , $\mathcal{C}_3^{\text{C-R}}(P)$, and $\mathcal{C}_4^{\text{C-R}}(P)$ are obtained as $(\text{CRB}_{2,\text{min}},R_{2,S})$, $(\text{CRB}_{3,\text{min}},R_{3,S})$, and $(\text{CRB}_{4,\text{min}},R_{4,S})$, respectively. 

	\section{Optimal Solution to Problem (P1) with Point Target}\label{section:point_target}
	
	In this section, we address problem (P1) with point target. 
	By defining $\tilde{\Gamma}_1 = \frac{2 \Gamma_1 L |\alpha|^{2}}{\sigma_s^2}$, problem (P1) is reformulated as
	\begin{subequations}\label{equ:point_target_formulation_eq}
		\begin{align}
			\text{(P1.1)}: \text{ }  \max _{\bm{Q} \succeq \bm{0}}  \text{ } & \log_2 \det \left(\bm{I}_{N_c} + \frac{1}{\sigma_c^2} \bm{H}_c \bm{Q} \bm{H}_c^H \right) \\
			\label{equ:CRLB_P_point_eq}
			\text { s.t.}  & \left[\begin{array}{cc}
				\operatorname{tr}\left(\dot{\mathbf{A}}^{H} \dot{\mathbf{A}} \bm{Q}\right)-\frac{1}{\tilde{\Gamma}_1} & \operatorname{tr}\left(\dot{\mathbf{A}}^{H} \mathbf{A} \bm{Q}\right) \\
				\operatorname{tr}\left(\mathbf{A}^{H} \dot{\mathbf{A}} \bm{Q}\right) & \operatorname{tr}\left(\mathbf{A}^{H} \mathbf{A} \bm{Q}\right)
			\end{array}\right] \succeq \bm{0} \\
			\label{equ:Power_P_point_eq}
			& \operatorname{tr}(\bm{Q}) \leq P,
		\end{align}
	\end{subequations}
	where (\ref{equ:CRLB_P_point_eq}) is equivalent to the CRB constraint in (\ref{equ:general_ex}) with $\text{CRB}_1(\bm{Q})$ given in (\ref{equ:point_target_CRLB}), and we omit $\theta$ for notational convenience. Problem (P1.1) is convex, since the objective function is concave and the constraints are convex. 
	To gain more insights, we use the Lagrange duality method to obtain a well-structured optimal solution to (P1.1).
	
	Let $\bm{Z}_P = \left[\alpha_d, \beta_d+j\gamma_d; \beta_d-j\gamma_d, \nu_d \right] \succeq \bm{0}$ and $\lambda \geq 0$ denote the dual varaibles associated with the constraints in  (\ref{equ:CRLB_P_point_eq}) and (\ref{equ:Power_P_point_eq}), respectively. The Lagrangian of (P1.1) is 
	\begin{align}
		\nonumber
		\mathcal{L}(\bm{Q},\lambda,\bm{Z}_P)  = & \log_2 \det \left(\bm{I}_{N_c} + \frac{1}{\sigma_c^2} \bm{H}_c \bm{Q} \bm{H}_c^H \right) - \lambda(\operatorname{tr}(\bm{Q})-P) + \\
		\nonumber
		&\operatorname{tr} \left( \left[\begin{array}{cc}
			\alpha_d & \beta_d+j\gamma_d \\
			\beta_d-j\gamma_d &\nu_d
		\end{array}\right]
		\left[\begin{array}{cc}
			\operatorname{tr}\left(\dot{\mathbf{A}}^{H} \dot{\mathbf{A}} \bm{Q}\right)-\frac{1}{\tilde{\Gamma}_1} & \operatorname{tr}\left(\dot{\mathbf{A}}^{H} \mathbf{A} \bm{Q}\right) \\
			\operatorname{tr}\left(\mathbf{A}^{H} \dot{\mathbf{A}} \bm{Q}\right) & \operatorname{tr}\left(\mathbf{A}^{H} \mathbf{A} \bm{Q}\right)
		\end{array}\right]\right)\\
		\nonumber
		= & \log_2 \det \left(\bm{I}_{N_c} + \frac{1}{\sigma_c^2} \bm{H}_c \bm{Q} \bm{H}_c^H \right) + \lambda P - \frac{\alpha_d}{\tilde{\Gamma}_1} - \operatorname{tr}\left( \bm{C}(\lambda, \bm{Z}_P) \bm{Q} \right),
	\end{align}
	where $\bm{C}(\lambda, \bm{Z}_P) \triangleq \lambda \bm{I} - (\alpha_d \dot{\mathbf{A}}^{H} \dot{\mathbf{A}} + (\beta_d+j\gamma_d) \dot{\mathbf{A}}^{H} \mathbf{A} + (\beta_d-j\gamma_d) \mathbf{A}^{H} \dot{\mathbf{A}} + \nu_d \mathbf{A}^{H} \mathbf{A})$. By using the fact that $\mathbf{a}^H \dot{\mathbf{a}} = 0$, $\mathbf{b}^H \dot{\mathbf{b}} = 0$, and $\dot{\mathbf{A}} = \mathbf{b} \dot{\mathbf{a}}^T + \dot{\mathbf{b}} \mathbf{a}^T$, $\bm{C}(\lambda,\bm{Z}_P)$ can be further simplified as
	\begin{align}\label{equ:C_simp}
		\nonumber
		\bm{C}(\lambda,\bm{Z}_P) = & \lambda \bm{I} - ( \alpha_d (\dot{\bm{a}}^* \dot{\bm{a}}^T \|\bm{b}\|_2^2 + \|\dot{\bm{b}}\|_2^2 \bm{a}^* \bm{a}^T) \\ 
		& + (\beta_d+j \gamma_d) \|\bm{b}\|_2^2 \dot{\bm{a}}^* \bm{a}^T + (\beta_d-j \gamma_d) \|\bm{b}\|_2^2 \bm{a}^* \dot{\bm{a}}^T + \nu_d \|\bm{b}\|_2^2 \bm{a}^* \bm{a}^T ),
	\end{align}
	with $\text{rank}(\bm{C}(\lambda,\bm{Z}_P) ) \geq M-2$.
	The Lagrange dual function of (P1.1) is
	\begin{align} \label{equ:dual_function_ex}
		g(\lambda,\bm{Z}_P) = \max_{\bm{Q} \succeq \bm{0}} \mathcal{L}(\bm{Q},\lambda,\bm{Z}_P).
	\end{align}
	The corresponding dual problem is
	\begin{align} \label{equ:dual_problem_ex}
		\text{(D1.1):} \text{ } \min_{\lambda \geq 0,\bm{Z}_P \succeq \bm{0}} g(\lambda, \bm{Z}_P). 
	\end{align}
	Notice that since problem (P1.1) is a convex optimization problem and meets the Slater's condition, the strong duality holds between (P1.1) and its dual problem (D1.1) \cite{boyd2004convex}. Therefore, primal problem (P1.1) can be solved by equivalently solving dual problem (D1.1) as follows. 
	
	\subsection{Finding Dual Function $g(\lambda,\bm{Z}_P)$ under Given $\lambda$ and $\bm{Z}_P$}
	
	First, consider problem (\ref{equ:dual_function_ex}) under given $\lambda$ and $\bm{Z}_P$,
	which is equivalent to solving the following problem by skipping the constant terms.
	\begin{align}\label{equ:find_dual_func_eq}
		\max_{\bm{Q} \succeq \bm{0}} \text{ } \log_2 \det \left(\bm{I}_{N_c} + \frac{1}{\sigma_c^2} \bm{H}_c \bm{Q} \bm{H}_c^H \right) - \operatorname{tr}(\bm{C}(\lambda, \bm{Z}_P) \bm{Q})
	\end{align}
	The eigenvalue decomposition (EVD) of $\bm{C}(\lambda, \bm{Z}_P)$ is given by
	\begin{align}\label{equ:C_evd}
		\bm{C}(\lambda, \bm{Z}_P) = \left[\begin{array}{cc}
			\underbrace{\bm{U}_1}_{M \times r_c} &  \underbrace{\bm{U}_0}_{M \times (M-r_c)}  
		\end{array}\right] \left[\begin{array}{cc}
		\bm{\Delta}_{r_c} & \bm{0}_{r_c \times (M-r_c)}  \\
		\bm{0}_{(M-r_c)\times r_c} & \bm{0}_{(M-r_c)\times(M-r_c)} 
		\end{array}\right] \left[\begin{array}{c}
			\bm{U}_1^H \\  \bm{U}_0^H  
		\end{array}\right] = \bm{U}_1 \bm{\Delta}_{r_c} \bm{U}_1^H,
	\end{align}
	where $r_c = \operatorname{rank}(\bm{C}(\lambda, \bm{Z}_P)) \geq M-2$ and $\bm{\Delta}_{r_c} = \operatorname{diag}(\lambda_{1}(\bm{C}(\lambda, \bm{Z}_P)),...,\lambda_{r_c}(\bm{C}(\lambda, \bm{Z}_P)))$. Without loss of generality, any feasible solution to problem (\ref{equ:find_dual_func_eq}) can be expressed as  
	\begin{align}\label{equ:Q_general}
		\nonumber
		\bm{Q} & = \left[\begin{array}{cc}
			\underbrace{\bm{U}_1}_{M \times r_c} &  \underbrace{\bm{U}_0}_{M \times (M-r_c)}  
		\end{array}\right] \left[\begin{array}{cc}
			\bm{Q}_{11} & \bm{Q}_{01}^H  \\
			\bm{Q}_{01} & \bm{Q}_{00} 
		\end{array}\right] \left[\begin{array}{c}
			\bm{U}_1^H \\  \bm{U}_0^H  
		\end{array}\right] 
		\\
		& = \underbrace{\bm{U}_1 \bm{Q}_{11} \bm{U}_1^H}_{\bm{Q}_{ll}} + \underbrace{\bm{U}_0 \bm{Q}_{01} \bm{U}_1^H + \bm{U}_1 \bm{Q}_{01}^H \bm{U}_0^H + \bm{U}_0 \bm{Q}_{00} \bm{U}_0^H }_{\bm{Q}_{\perp}}.
	\end{align}
	Notice that both $\bm{Q}_{ll}$ and $\bm{Q}_{\perp}$ are hermitian. It is then easy to see that $\operatorname{tr}(\bm{C}(\lambda, \bm{Z}_P)\bm{Q}) = \operatorname{tr}(\bm{Q}_{11} \bm{\Delta}_{r_c})$. Recall that the SVD of $\bm{H}_c$ is $\bm{H}_c = \bm{U}_c \bm{\Sigma}_c \bm{V}_c^H = \left[\bm{U}_{c1}, \bm{U}_{c0}\right] \bm{\Sigma}_c \left[\bm{V}_{c1}, \bm{V}_{c0}\right]^H$.
	We then have the following lemma. 
	
	
	
	\begin{lemma}\label{lemma:orthogonal_condition}
		\emph{In order for the optimal value of problem (\ref{equ:find_dual_func_eq}) or equivalently  $g(\lambda,\bm{Z}_P)$ to be bounded from above, it must hold that $\bm{C}(\lambda,\bm{Z}_P) \succeq \bm{0}$ and $\mathcal{R}(\bm{V}_{c1}) \subseteq \mathcal{R}(\bm{U}_{1})$.}
		\begin{proof}
			We prove this lemma by contradiction. First, suppose that $\bm{C}(\lambda,\bm{Z}_P)$ is not positive semi-definite, and let $\bm{v}_{-}$ denote the eigenvector corresponding to any one negative eigenvalue. In this case, we can choose $\bm{Q} = \xi_{-} \bm{v}_{-} \bm{v}_{-}^H$ with $\xi_{-} > 0$. By setting $\xi_{-} \rightarrow \infty$, we accordingly have $g(\lambda,\bm{Z}_P) \rightarrow \infty$, thus resulting in a contradiction. Next, suppose that $\mathcal{R}(\bm{V}_{c1}) \subseteq \mathcal{R}(\bm{U}_{1})$ does not hold, and thus
			there exists a vector $\bm{v}_c$ such that $\bm{v}_c \in \mathcal{R}(\bm{V}_{c1})$ and $\bm{v}_c \notin \mathcal{R}(\bm{U}_1)$. As $\mathcal{R}(\bm{U}_0) \oplus \mathcal{R}(\bm{U}_1) = \mathbb{C}^M$, we have $\bm{v}_c \in \mathcal{R}(\bm{U}_0)$. In this case, we set $\bm{Q}_{11} = \bm{Q}_{01} = \bm{0}$, and $\bm{U}_0 \bm{Q}_{00} \bm{U}_0^H =  \xi_c \bm{v}_{c} \bm{v}_{c}^H $ or equivalently $\bm{Q}_{00} = \xi_c \bm{U}_0^H \bm{v}_{c} \bm{v}_{c}^H \bm{U}_0$. By letting $\xi_c \rightarrow \infty$, we have $g(\lambda,\bm{Z}_P) \rightarrow \infty$, which is a contradiction again. This completes the proof.
		\end{proof}
	\end{lemma}
	
	According to Lemma \ref{lemma:orthogonal_condition}, we only need to deal with problem (\ref{equ:find_dual_func_eq}) in the case with $\bm{C}(\lambda,\bm{Z}_P) \succeq \bm{0}$ and $\mathcal{R}(\bm{V}_{c1}) \subseteq \mathcal{R}(\bm{U}_{1})$. In this case,
	problem (\ref{equ:find_dual_func_eq}) is simplified as the optimization of $\bm{Q}_{11}$:
	\begin{align}\label{equ:Q_11_formulation}
		\max_{\bm{Q}_{11} \succeq \bm{0}} \text{ } \log_2 \det \left(\bm{I}_{N_c} + \frac{1}{\sigma_c^2} \bm{H}_c \bm{U}_1 \bm{Q}_{11} \bm{U}_1^H \bm{H}_c^H \right) - \operatorname{tr}(\bm{Q}_{11} \bm{\Delta}_{r_c} ).  
	\end{align}
	Let $\tilde{\bm{Q}}_{11} \triangleq \bm{\Delta}_{r_c}^{\frac{1}{2}} \bm{Q}_{11} \bm{\Delta}_{r_c}^{\frac{1}{2}}$. Problem (\ref{equ:Q_11_formulation}) is transformed as
	\begin{align}\label{equ:Q_11_tilde_formulation}
		\max_{\tilde{\bm{Q}}_{11} \succeq \bm{0}} \text{ } \log_2 \det \left(\bm{I}_{N_c} + \frac{1}{\sigma_c^2} \bm{H}_c \bm{U}_1  \bm{\Delta}_{r_c}^{-\frac{1}{2}} \tilde{\bm{Q}}_{11} \bm{\Delta}_{r_c}^{-\frac{1}{2}}  \bm{U}_1^H \bm{H}_c^H \right) - \operatorname{tr}(\tilde{\bm{Q}}_{11}).
	\end{align}
	Let the SVD of $\bm{W} \triangleq \bm{H}_c \bm{U}_1  \bm{\Delta}_{r_c}^{-\frac{1}{2}} \in \mathbb{C}^{N_c \times r_c}$ be expressed as $\bm{U}_W \bm{\Sigma}_W \bm{V}_W^H$. Then, we introduce 
	$\bar{\bm{Q}}_{11} \triangleq \bm{V}_W^H \tilde{\bm{Q}}_{11} \bm{V}_W$, and accordingly reformulate problem (\ref{equ:Q_11_tilde_formulation}) as
	\begin{align}\label{equ:P0_dual_trans_22}
		\max_{\bar{\bm{Q}}_{11} \succeq \bm{0}} \text{ } \log_2 \det \left(\bm{I}_{r_c} + \frac{1}{\sigma_c^2} \bm{\Sigma}_W^2 \bar{\bm{Q}}_{11} \right) - \operatorname{tr}(\bar{\bm{Q}}_{11}),
	\end{align}
	where $\bm{\Sigma}_W^2 \triangleq \bm{\Sigma}_W^H \bm{\Sigma}_W = \operatorname{diag}(\zeta_{1}^2(\bm{W}),...,\zeta_{r_c}^2(\bm{W}))$. By applying the  Hadamard's inequality \cite{horn2012matrix}, it follows that the optimal solution to problem (\ref{equ:P0_dual_trans_22}) is diagonal, i.e., $\bar{\bm{Q}}_{11}^* = \operatorname{diag}(\bar{p}_{1,1}^*,...,\bar{p}_{1,r_c}^*)$. By further applying the  Karush-Kuhn-Tucker (KKT) conditions, 
	the optimal $\{\bar{p}_{1,k}^*\}$ follows the water-filling-like structure, i.e.,
	\begin{align}\label{equ:optimal_power_allocation_2}
		\bar{p}_{1,k}^* = \left(\frac{1}{\ln 2}-\frac{\sigma_c^2}{\zeta_{k}^2(\bm{W})}\right)^+, \forall k \in \{1,2,...,r_c\}.
	\end{align}
	The optimal solution to problem (\ref{equ:Q_11_formulation}) is thus expressed as
	\begin{align} \label{equ:optimal_Q_point_2}
		\bm{Q}_{11}^* = \bm{\Delta}_{r_c}^{-\frac{1}{2}} \bm{V}_W \operatorname{diag}(\bar{p}_{1,1}^*,...,\bar{p}_{1,r_c}^*) \bm{V}_W^H \bm{\Delta}_{r_c}^{-\frac{1}{2}}.
	\end{align}
	Based on Lemma \ref{lemma:orthogonal_condition}, it is clear that $g(\lambda, \bm{Z}_P)$ only depends on $\bm{Q}_{ll}$ or $\bm{Q}_{11}^*$ in (\ref{equ:Q_general}) and we thus have $g(\lambda, \bm{Z}_P) = \mathcal{L}(\bm{Q}^*,\lambda,\bm{Z}_P)$ with $\bm{Q}^* = \bm{U}_1 \bm{Q}_{11}^* \bm{U}_1^H$. Note that $\bm{Q}^*$ is non-unique in general.
	
	
	
	\subsection{Optimal Solution to (D1.1)}
	
	Next, we solve dual problem (D1.1). Based on Lemma \ref{lemma:orthogonal_condition}, (D1.1) is further reexpressed as
	\begin{align} \label{equ:dual_problem_reex}
		\text{(D1.1):} \text{ } \min_{\lambda \geq 0,\bm{Z}_P \succeq \bm{0}} g(\lambda, \bm{Z}_P), \quad \text{s.t. } \bm{C}(\lambda,\bm{Z}_P) \succeq \bm{0}.
	\end{align}
	As problem (D1.1) has a convex but in general non-differentiable objective function with linear matrix inequality constraints, it is convex and can be solved optimally by applying subgradient-based methods, e.g., the ellipsoid method \cite{boyd2004convex}. Towards this end, we need to obtain the subgradients of the objective function and the constraint functions in (D1.1). First, for the objective function $g(\lambda,\bm{Z}_P)$, the subgradient $\partial g$ at $\left(\lambda, \alpha_d, \beta_d, \gamma_d, \nu_d\right)$ is given by
	\begin{align}\label{equ:subgradient_obj}
		\left[P-\operatorname{tr}(\bm{Q}^*), \operatorname{tr}(\dot{\mathbf{A}}^{H} \dot{\mathbf{A}} \bm{Q}^*) - \frac{1}{\tilde{\Gamma}_1}, \operatorname{tr}((\dot{\mathbf{A}}^{H} \mathbf{A} + \mathbf{A}^H \dot{\mathbf{A}}) \bm{Q}^* ), \operatorname{tr}(j(\dot{\mathbf{A}}^{H} \mathbf{A} - \mathbf{A}^H \dot{\mathbf{A}}) \bm{Q}^* ), \operatorname{tr}(\mathbf{A}^H \mathbf{A} \bm{Q}^*) \right]^T,
	\end{align}
	where $\bm{Q}^*$ denotes the optimal solution to $\max_{\bm{Q} \succeq \bm{0}} \mathcal{L}(\bm{Q},\lambda,\bm{Z}_P)$ under given $\lambda$ and $\bm{Z}_P$.
	Next, given fixed $\left(\lambda, \alpha_d, \beta_d, \gamma_d, \nu_d\right)$, we denote the eigenvector of $\bm{Z}_P$ corresponding to its minimum eigenvalue as $\bm{z} = \left[z_1,z_2\right]^T$ and that of $\bm{C}(\lambda,\bm{Z}_P)$ as $\bm{q}$. 
	The constraints $\bm{Z}_P \succeq \bm{0}$ and $\bm{C}(\lambda,\bm{Z}_P) \succeq \bm{0}$ are thus equivalent to $\bm{z}^H \bm{Z}_P \bm{z} \geq 0$ and $\bm{q}^H \bm{C}(\lambda,\bm{Z}_P) \bm{q} \geq 0$, respectively. Therefore, for constraint function 
	$\bm{z}^H \bm{Z}_P \bm{z}$, the subgradient is $\left[0,|z_1|^2, \bar{z}_1 z_2 + z_1 \bar{z}_2, -j z_1 \bar{z}_2 + j \bar{z}_1 z_2, |z_2|^2 \right]^T$, and for constraint function $\bm{q}^H \bm{C}(\lambda,\bm{Z}_P) \bm{q}$, the subgradient is \cite{boyd2004convex}
	\begin{align}\label{subgradient_const_func3}
		\left[\bm{q}^H \bm{q}, -\bm{q}^H \dot{\mathbf{A}}^{H} \dot{\mathbf{A}} \bm{q}, -\bm{q}^H(\dot{\mathbf{A}}^{H} \mathbf{A} + \mathbf{A}^H \dot{\mathbf{A}})\bm{q}, \bm{q}^H(-j\dot{\mathbf{A}}^{H} \mathbf{A} + j \mathbf{A}^H \dot{\mathbf{A}})\bm{q}, -\bm{q}^H \mathbf{A}^H \mathbf{A} \bm{q}\right]^T.
	\end{align}
	
	By implementing the ellipsoid method based on the obtained subgradients, the optimal solution to (D1.1), denoted by $\lambda^{\text{opt}}$ and $\bm{Z}_P^{\text{opt}}$, is obtained.
	
	\subsection{Optimal Solution to Primal Problem (P1.1) or (P1)}
	
	Finally, with $\lambda^{\text{opt}}$ and $\bm{Z}_P^{\text{opt}}$ at hand, we obtain the optimal primal solution $\bm{Q}_1^{\text{opt}}$ to (P1.1). Based on (\ref{equ:optimal_Q_point_2}), the optimal solution of $\bm{Q}_{11}$ to problem (\ref{equ:dual_function_ex}) or (\ref{equ:find_dual_func_eq}) under  $\lambda^{\text{opt}}$ and $\bm{Z}_P^{\text{opt}}$ is 
	%
	\begin{align}\label{eq:optimal_sol_P11}
		\bm{Q}_{11}^{\text{opt}} = \{\bm{\Delta}_{r_c}^{\text{opt}}\}^{-\frac{1}{2}} \bm{V}_W^{\text{opt}} \operatorname{diag}(\bar{p}_{1,1}^{\text{opt}},...,\bar{p}_{1,r_c}^{\text{opt}}) \{\bm{V}_W^{\text{opt}}\}^H \{\bm{\Delta}_{r_c}^{\text{opt}}\}^{-\frac{1}{2}},
	\end{align}
	where $\bm{\Delta}_{r_c}^{\text{opt}}$, $\bm{V}_W^{\text{opt}}$, and $(\bar{p}_{1,1}^{\text{opt}},...,\bar{p}_{1,r_c}^{\text{opt}})$ are decided based on $\bm{C}(\lambda^{\text{opt}},\bm{Z}_P^{\text{opt}}) = \bm{U}_1^{\text{opt}} \bm{\Delta}_{r_c}^{\text{opt}} \{\bm{U}_1^{\text{opt}}\}^H$ via (\ref{equ:optimal_Q_point_2}). Note that if $\bm{C}(\lambda^{\text{opt}},\bm{Z}_P^{\text{opt}})\succ \bm{0}$, we have $\bm{Q}_1^{\text{opt}} = \bm{Q}_{ll}^{\text{opt}} = \bm{U}_1^{\text{opt}} \bm{Q}_{11}^{\text{opt}} \{\bm{U}_1^{\text{opt}}\}^H$, or equivalently,
	\begin{align}\label{equ:Q_opt_C_PD}
		\bm{Q}_1^{\text{opt}} = \bm{C}(\lambda^{\text{opt}},\bm{Z}_P^{\text{opt}})^{-1/2} \bm{V}_G^{\text{opt}} \operatorname{diag}(p_{1,1}^{\text{opt}},...,p_{1,M}^{\text{opt}}) \{\bm{V}_G^{\text{opt}}\}^H \bm{C}(\lambda^{\text{opt}},\bm{Z}_P^{\text{opt}})^{-1/2},
	\end{align} 
	where $\bm{V}_G^{\text{opt}}$ contains the right singular vectors of $\bm{G}^{\text{opt}} \triangleq \bm{H}_c \bm{C}(\lambda^{\text{opt}},\bm{Z}_P^{\text{opt}})^{-1/2}$, i.e., $\bm{G}^{\text{opt}} = \bm{U}_G^{\text{opt}} \bm{\Sigma}_G^{\text{opt}} \{\bm{V}_G^{\text{opt}}\}^H$ with 
	$\{\bm{\Sigma}_G^{\text{opt}}\}^H \bm{\Sigma}_G^{\text{opt}}= \operatorname{diag}(\zeta_{1}^2(\bm{G}^{\text{opt}}),...,\zeta_{M}^2(\bm{G}^{\text{opt}}))$ and $p_{1,k}^{\text{opt}} = \left(\frac{1}{\ln 2}-\frac{\sigma_c^2}{\zeta_{k}^2(\bm{G}^{\text{opt}})}\right)^+$, $\forall k \in \{1,...,M\}$. 
	However, if  $\bm{C}(\lambda^{\text{opt}},\bm{Z}_P^{\text{opt}}) \succeq \bm{0}$ is not positive definite, we need to further determine $\bm{Q}^{\text{opt}}_{01}, \bm{Q}^{\text{opt}}_{00}$ under given $\bm{Q}_{ll}^{\text{opt}}$ by solving the following feasibility problem: 
	\begin{align}
		\text{Find }  \{\bm{Q}_{01}, \bm{Q}_{00}\}, \quad \text{s.t. }  (\text{\ref{equ:CRLB_P_point_eq}}), (\text{\ref{equ:Power_P_point_eq}}), (\text{\ref{equ:Q_general}}).
	\end{align}
	In this case, we have the optimal primal solution to problem (P1) as $\bm{Q}_1^{\text{opt}} = \bm{Q}_{ll}^{\text{opt}} + \bm{Q}_{\perp}^{\text{opt}}$, where $\bm{Q}_{\perp}^{\text{opt}} = \bm{U}_0^{\text{opt}} \bm{Q}_{01}^{\text{opt}} \{\bm{U}_1^{\text{opt}}\}^H + \bm{U}_1^{\text{opt}} \{\bm{Q}^{\text{opt}}_{01}\}^H \{\bm{U}_0^{\text{opt}}\}^H + \bm{U}_0^{\text{opt}} \bm{Q}^{\text{opt}}_{00} \{\bm{U}_0^{\text{opt}}\}^H$.
	
	Notice that $\bm{Q}_{ll}^{\text{opt}}$ is used for both sensing and communication in general and $\bm{Q}_{\perp}^{\text{opt}}$ is used for sensing only. However, in most practical cases with random target directions and communication channels, $\mathcal{R}(\bm{V}_{c1}) \subseteq \mathcal{R}(\bm{U}_{1})$ holds with $\mathcal{R}(\bm{U}_1) = \mathbb{C}^M$. In this case, we have $\bm{C}(\lambda^{\text{opt}}, \bm{Z}^{\text{opt}}) \succ \bm{0}$ according to Lemma \ref{lemma:orthogonal_condition} and thus  $\bm{Q}_1^{\text{opt}}$ can be directly obtained from (\ref{equ:Q_opt_C_PD}). As a result, it is interesting to see from (\ref{equ:Q_opt_C_PD}) that the optimal solution is obtained by first implementing SVD to diagonalize the composite channel $\bm{H}_c \bm{C}(\lambda^{\text{opt}},\bm{Z}_P^{\text{opt}})^{-1/2}$, followed by the water-filling-like power allocation over the decomposed subchannels.

	\section{Optimal Solutions to Problems (P2)-(P4) with Extended Target}\label{section:extended_target}
	
	This section addresses problems (P2)-(P4)  for Scenarios 2-4 with extended target to find the whole Pareto boundary of $\mathcal{C}_2^{\text{C-R}}(P)$, $\mathcal{C}_3^{\text{C-R}}(P)$, and $\mathcal{C}_4^{\text{C-R}}(P)$.
	
	\subsection{Optimal Solution to Problem (P2) with Trace-CRB}
	
	First, we consider problem (P2). By defining $\tilde{\Gamma}_2 \triangleq \frac{L \Gamma_2}{\sigma_s^2 N_s}$, problem (P2) is re-expressed as 
	\begin{align}\label{equ:P_1}
		\text{(P2.1)}: \text{ } \max _{\bm{Q}\succeq \bm{0}}  \text{ } \log_2 \det \left(\bm{I}_{N_c} + \frac{1}{\sigma_c^2} \bm{H}_c \bm{Q} \bm{H}_c^H \right), \quad \text { s.t. }  \operatorname{tr}(\bm{Q}^{-1}) \leq \tilde{\Gamma}_2, \quad \operatorname{tr}(\bm{Q}) \leq P.
	\end{align}
	Note that problem (P2.1) is convex. To solve this problem, we first define $\tilde{\bm{Q}} \triangleq \mv{V}_c^H \mv{Q} \mv{V}_c$, where $\bm{V}_c$ stems from the SVD of $\bm{H}_c$ given by $\bm{H}_c = \mv{U}_c \mv{\Sigma}_c \mv{V}_c^H$.
	Accordingly, (P2.1) is equivalently reformulated as
	\begin{align}\label{equ:P1_equivalent1}
		\text{(P2.2):} \text{ } \max _{\tilde{\bm{Q}}\succeq \bm{0}}  \text{ } \log_2 \det \left(\mv{I}_{M} + \frac{1}{\sigma_c^2}  \mv{\Sigma}_c^2 \tilde{\bm{Q}} \right), \quad
		\text { s.t. }  \operatorname{tr}( \tilde{\bm{Q}}^{-1}) \leq  \tilde{\Gamma}_2, \quad
		\operatorname{tr}(\tilde{\bm{Q}}) \leq P,
	\end{align}
	where $\mv{\Sigma}_c^2 \triangleq \mv{\Sigma}_c^H \mv{\Sigma}_c = \operatorname{diag}(\zeta_{1}^2(\bm{H}_c),...,\zeta_{r}^2(\bm{H}_c),0,...,0) \in \mathbb{R}^{M \times M}$. Here, the objective function in (\ref{equ:P1_equivalent1}) is obtained based on  $\det (\bm{I}_{N_c} + \frac{1}{\sigma_c^2} \bm{H}_c \bm{Q} \bm{H}_c^H)  = \det (\mv{I}_{M} + \frac{1}{\sigma_c^2}  \mv{\Sigma}_c^2 \tilde{\bm{Q}} )$, and the constraints in (\ref{equ:P1_equivalent1}) follow from the constraints in (\ref{equ:P_1}) since  $\operatorname{tr}({\bm{Q}}^{-1}) = \operatorname{tr}((\mv{V}_c \tilde{\mv{Q}} \mv{V}_c^H)^{-1}) = \operatorname{tr}(\mv{V}_c {\tilde{\mv{Q}}}^{-1} \mv{V}_c^{H}) = \operatorname{tr}(\mv{V}_c^H \mv{V}_c {\tilde{\mv{Q}}}^{-1} ) = \operatorname{tr}({\tilde{\mv{Q}}}^{-1} )$ and $\operatorname{tr}({\bm{Q}}) = \operatorname{tr}( \mv{V}_c \tilde{\mv{Q}} \mv{V}_c^H) = \operatorname{tr}(\mv{V}_c^H \mv{V}_c \tilde{\mv{Q}} ) = \operatorname{tr}(\tilde{\mv{Q}})$.
	Next, we have the following proposition.
	\begin{proposition}\label{Pro:diagonal_optimal}
		\emph{The optimal solution to problem (P2.2) is a diagonal matrix with strictly positive diagonal elements, i.e., $\tilde{\bm{Q}}_2 = \operatorname{diag}(p_{2,1},p_{2,2},...,p_{2,M})$, where $p_{2,k} > 0, \forall k \in \{1,\ldots,M\}$.}
	\end{proposition}
	\begin{proof} 
		See Appendix \ref{Trace_diag_optimal_proof}.
	\end{proof}
	Based on Proposition \ref{Pro:diagonal_optimal}, problem (P2.2) is equivalently reformulated as
	\begin{align}\label{equ:P_1P}
		\text{(P2.3):} \max _{\{p_{2,k} \ge 0\}}  \sum_{k=1}^r \log_2 \left(1+\frac{\zeta_k^2(\bm{H}_c) p_{2,k}}{\sigma_c^2} \right), \quad
		\text { s.t. }  \sum_{k=1}^M \frac{1}{p_{2,k}} \leq  \tilde{\Gamma}_2, \quad \sum_{k=1}^M p_{2,k} \leq P. 
	\end{align}
	\begin{proposition}\label{pro:prime_dual_relationship}
		\emph{The optimal power allocation solution to (P2.3) is obtained as
			\begin{align}\label{equ:P_2_3_power}
			p_{2,k}^{\text{opt}} = 
			\begin{cases}
				-t_{1,k}^{\text{opt}} + \sqrt[3]{-t_{2,k}^{\text{opt}}+\sqrt{(t_{2,k}^{\text{opt}})^2+(t_{3,k}^{\text{opt}})^3}} +  \sqrt[3]{-t_{2,k}^{\text{opt}}-\sqrt{(t_{2,k}^{\text{opt}})^2+(t_{3,k}^{\text{opt}})^3}}, \quad & \forall k \in \{1,\ldots, r\},	\\
				 \sqrt{\mu^{\text{opt}}_2/v^{\text{opt}}_2},  \quad  \forall k \in \{r+1,\ldots, M\}. &
			\end{cases}	
		\end{align}
			where
			$
			t_{1,k}^{\text{opt}} = b_k^{\text{opt}}/(3a^{\text{opt}})$, $t_{2,k}^{\text{opt}} = (27(a^{\text{opt}})^2d_k^{\text{opt}}-9a^{\text{opt}}b_k^{\text{opt}}c^{\text{opt}}+2(b_k^{\text{opt}})^3)/(54(a^{\text{opt}})^3)$,  $t_{3,k}^{\text{opt}} = (3a^{\text{opt}}c^{\text{opt}}-(b_k^{\text{opt}})^2)/(9(a^{\text{opt}})^2)$
			with $a^{\text{opt}} = v^{\text{opt}}_2, b_k^{\text{opt}} = v^{\text{opt}}_2 \frac{\sigma_c^2}{\zeta_k^2(\bm{H}_c)} - \frac{1}{\text{ln2}}, c^{\text{opt}} = -\mu^{\text{opt}}_2$, and $d_k^{\text{opt}} = -\mu^{\text{opt}}_2 \frac{\sigma_c^2}{\zeta_k^2(\bm{H}_c)}$. Here, $\mu^{\text{opt}}_2$ and $v^{\text{opt}}_2$ are the optimal dual variables associated with the CRB constraint and the power constraint in (\ref{equ:P_1P}), respectively.
	}\end{proposition}
	\begin{proof}
		See Appendix \ref{Proof:prime_dual_relationship}.		
	\end{proof}

	Finally, combining Propositions \ref{Pro:diagonal_optimal} and \ref{pro:prime_dual_relationship}, the optimal solution to (P2) is obtained as $\bm{Q}_2^{\text{opt}} = \bm{V}_c \tilde{\bm{Q}}_2^{\text{opt}}\bm{V}_c^H,$
	where $\tilde{\bm{Q}}_2^{\text{opt}} = \operatorname{diag}(p_{2,1}^{\text{opt}},\ldots, p_{2,M}^{\text{opt}})$, with $\{p_{2,k}^{\text{opt}}\}$ given in Proposition \ref{pro:prime_dual_relationship}.
	

	\subsection{Optimal Solution to Problem (P3) with MaxEig-CRB}
	
	Next, we consider problem (P3) with MaxEig-CRB in Scenario 3, which is re-expressed as
	\begin{align} \label{equ:P_extended_eig_opt}
		\text{(P3)}: \text{ } \max _{\bm{Q} \succeq \bm{0} }  \text{ } \log_2 \det \left(\bm{I}_{N_c} + \frac{1}{\sigma_c^2} \bm{H}_c \bm{Q} \bm{H}_c^H \right), 
		\quad \text {s.t. }  \lambda_{\text{max}} (\overline{\bold{CRB}}(\bm{Q})) \leq \Gamma_3, \quad \operatorname{tr}(\bm{Q}) \leq P.
	\end{align}
	Notice that $\overline{\bold{CRB}}(\bm{Q})^{-1} = \bm{J}(\bm{Q}) = \frac{L}{\sigma_s^2} \bm{Q}^T \otimes \bm{I}_{N_s}$, and as a result, $\lambda_{\text{max}} (\overline{\bold{CRB}}(\bm{Q})) \leq \Gamma_3$ is equivalent to $\lambda_{\text{min}} (\bm{J}(\bm{Q})) \geq \frac{1}{\Gamma_3}$. Furthermore, according to the eigenvalue property of Kronecker product, i.e.,  $\bm{\lambda}(\frac{L}{\sigma_s^2} \bm{Q}^T \otimes \bm{I}_{N_s}) = \{\beta_i \gamma_j: \beta_i \in \bm{\lambda}(\frac{L}{\sigma_s^2} \bm{Q}^T), \gamma_j \in \bm{\lambda}(\bm{I}_{N_s})\}$,
	the minimum eigenvalue of $\bm{J}(\bm{Q})$ is equivalent to the minimum eigenvalue of $\frac{L}{\sigma_s^2} \bm{Q}^T$. Thus, problem (P3) is equivalently reformulated as
	\begin{align}\label{equ:P_extended_eig_opt_eq}
		\text{(P3.1)}: \text{ } \max _{\tilde{\bm{Q}}\succeq \bm{0}}  \text{ } \log_2 \det \left(\mv{I}_{M} + \frac{1}{\sigma_c^2}  \mv{\Sigma}_c^2 \tilde{\bm{Q}} \right),
		\quad \text { s.t. }  \tilde{\bm{Q}} \succeq  \frac{\sigma_s^2 }{L \Gamma_3} \bm{I}_{M}, \quad
		\operatorname{tr}(\tilde{\bm{Q}}) \leq P,
	\end{align}
	where $\tilde{\bm{Q}} \triangleq \bm{V}_c^H \bm{Q} \bm{V}_c$. We have the following proposition.
	\begin{proposition}\label{Pro:diagonal_optimal_P3}
		\emph{The optimal solution to problem (P3.1) is a diagonal matrix with strictly positive diagonal elements, i.e., $\tilde{\bm{Q}}_3 = \operatorname{diag}(p_{3,1},p_{3,2},...,p_{3,M})$, where $p_{3,k} > 0, \forall k \in \{1,...,M\}$.}
	\end{proposition}
	\begin{proof}
		See Appendix \ref{Proof:prop_Eig_diag_optimal}.
	\end{proof}
	Based on Proposition \ref{Pro:diagonal_optimal_P3} and defining $\tilde{\Gamma}_e \triangleq \frac{\sigma_s^2 }{L \Gamma_3}$, we further simplify problem (P3.1) as
	\begin{align} 
		\label{equ:P_extended_eig_opt_diag}
		\text{(P3.2)}:  \max _{\{p_{3,k} \ge 0\}}  \text{ } \sum_{k=1}^r \log_2 \left(1+\frac{\zeta_k^2(\bm{H}_c) p_{3,k}}{\sigma_c^2} \right), 
		\text { s.t. }  p_{3,k} \geq \tilde{\Gamma}_e, \forall k \in \{1,...,M\}, \text{ }
		\sum_{k=1}^M p_{3,k} \leq P.
	\end{align}
	\begin{proposition}\label{pro:prime_dual_relationship_eig}
		\emph{The optimal power allocation solution of $\{p_{3,k}^{\text{opt}}\}_{k=1}^M$ to problem (P3.2) is 
			\begin{align}
				\label{equ:Lag_zero1_in_pro_eig}
				p_{3,k}^{\text{opt}} = \begin{cases}
					\max \{\frac{1}{v^{\text{opt}}_3 \ln 2} - \frac{\sigma_c^2}{\zeta_k^2(\bm{H}_c)}, \tilde{\Gamma}_e\}
					, \quad & \forall k \in \{1,\ldots, r\},	\\
					\tilde{\Gamma}_e, \quad & \forall k \in \{r+1,\ldots, M\}.
				\end{cases}	
			\end{align}
			Here, $v^{\text{opt}}_3$ denotes the optimal dual variable associated with the power constraint in (\ref{equ:P_extended_eig_opt_diag}).
	}\end{proposition}
	\begin{proof}
		See Appendix \ref{Proof:prop_Eig_semi_form}.
	\end{proof}
	Finally, with Propositions \ref{Pro:diagonal_optimal_P3} and \ref{pro:prime_dual_relationship_eig}, the optimal solution to (P3) is obtained as $\bm{Q}_3^{\text{opt}} = \bm{V}_c \tilde{\bm{Q}}_3^{\text{opt}}\bm{V}_c^H$, where $\tilde{\bm{Q}}_3^{\text{opt}} = \operatorname{diag}(p_{3,1}^{\text{opt}},\ldots, p_{3,M}^{\text{opt}})$, with $\{p_{3,k}^{\text{opt}}\}$ given in Proposition \ref{pro:prime_dual_relationship_eig}.
	
	
	\subsection{Optimal Solution to Problem (P4) with Det-CRB}
	
	Next, we consider problem (P4) with Det-CRB in Scenario 4, which is re-expressed as
	\begin{align}\label{equ:P_extended_det_opt_same}
		\text{(P4.1)}: \text{ } \max _{\tilde{\bm{Q}}\succeq \bm{0}}  \text{ } \log_2 \det \left(\mv{I}_{M} + \frac{1}{\sigma_c^2}  \mv{\Sigma}_c^2 \tilde{\bm{Q}} \right), \quad
		\text { s.t. } \ln \det (\tilde{\bm{Q}}) \geq \tilde{\Gamma}_d, \quad
		\operatorname{tr}(\tilde{\bm{Q}}) \leq P,
	\end{align}
	where $\tilde{\Gamma}_d \triangleq M \ln \frac{\sigma_s^2}{L} - \frac{1}{N_s} \ln \Gamma_4$ and $\tilde{\bm{Q}} \triangleq \mv{V}_c^H \mv{Q} \mv{V}_c$. We have the following proposition.
	\begin{proposition}\label{Pro:diagonal_optimal_P4}
		\emph{The optimal solution to problem (P4.1) is a diagonal matrix with strictly positive diagonal elements, i.e., $\tilde{\bm{Q}}_4 = \operatorname{diag}(p_{4,1},p_{4,2},...,p_{4,M})$, where $p_{4,k} > 0, \forall k \in \{1,\ldots,M\}$.}
	\end{proposition}
	\begin{proof}
		This proposition can be verified by applying the Hadamard's inequality \cite{horn2012matrix} to the objective function and the CRB constraint in  (\ref{equ:P_extended_det_opt_same}), for which the details are omitted.
	\end{proof}
	
	Based on Proposition \ref{Pro:diagonal_optimal_P4}, (P4.1) is further simplified as
	\begin{align}\label{equ:P_extended_det_opt_diag}
		\text{(P4.2)}: \text{ } \max _{\{p_{4,k} \ge 0\}} \text{ } \sum_{k=1}^r \log_2 \left(1+\frac{\zeta_k^2(\bm{H}_c) p_{4,k}}{\sigma_c^2} \right), \quad
		\text { s.t. }  \sum_{k=1}^{M} \ln p_{4,k} \geq \tilde{\Gamma}_d, \quad
		\sum_{k=1}^M p_{4,k} \leq P.
	\end{align}
	\begin{proposition}\label{pro:prime_dual_relationship_det}
		\emph{The optimal solution to (P4.2) is obtained as
			\begin{align}
				p_{4,k}^{\text{opt}} = 
				\begin{cases}
					\frac{-(v^{\text{opt}}_4\frac{\sigma_c^2}{\zeta_k^2(\bm{H}_c)}-\mu^{\text{opt}}_4-\frac{1}{\ln 2})+\sqrt{(v^{\text{opt}}_4\frac{\sigma_c^2}{\zeta_k^2(\bm{H}_c)}-\mu^{\text{opt}}_4-\frac{1}{\ln 2})^2+4v^{\text{opt}}_4\mu^{\text{opt}}_4 \frac{\sigma_c^2}{\zeta_k^2(\bm{H}_c)} }}{2 v^{\text{opt}}_4}, \quad & \forall k \in \{1,\ldots, r\},	\\
					\mu^{\text{opt}}_4/v^{\text{opt}}_4,  \quad & \forall k \in \{r+1,\ldots, M\}. 
				\end{cases}	
			\end{align}
%
			Here, $\mu^{\text{opt}}_4$ and $v^{\text{opt}}_4$ are the optimal dual variables associated with the CRB constraint and the power constraint in (\ref{equ:P_extended_det_opt_diag}), respectively.
	}\end{proposition}
	\begin{proof}
		The optimal solution to (P4.2) is obtained by the Lagrange duality method. The derivation is similar to that in Appendix \ref{Proof:prime_dual_relationship}, for which the details are omitted for brevity.
	\end{proof}
	
	Finally, with Propositions \ref{Pro:diagonal_optimal_P4} and \ref{pro:prime_dual_relationship_det}, the optimal solution to (P4) is obtained as $\bm{Q}_4^{\text{opt}} = \bm{V}_c \tilde{\bm{Q}}_4^{\text{opt}}\bm{V}_c^H$, where $\tilde{\bm{Q}}_4^{\text{opt}} = \operatorname{diag}(p_{4,1}^{\text{opt}},\ldots, p_{4,M}^{\text{opt}})$, with $\{p_{4,k}^{\text{opt}}\}$ given in Proposition \ref{pro:prime_dual_relationship_det}.
	
	\subsection{Optimal Solution Structures} \label{section:sol_structure_extended}
	
	To gain more insights, this subsection discusses the structure of the optimal transmit covariance solution $\bm{Q}_i^{\text{opt}}$'s for Scenarios 2-4 under Trace-CRB, MaxEig-CRB, and Det-CRB, respectively. In particular, we express $\bm{V}_c$ as $\bm{V}_c = [\bar{\bm{V}}_{c}, \hat{\bm{V}}_c]$, where $\bar{\bm{V}}_{c} \in \mathbb{C}^{M\times r}$ consists of the first $r$ right singular vectors of $\bm{H}_c$, and $\hat{\bm{V}}_{c}\in \mathbb{C}^{M\times (M-r)}$ consists of the other $M-r$ ones. It thus follows that $\bm{Q}_i^{\text{opt}}$ for Scenario $i \in \{2,3,4\}$ can be equivalently written as
	\begin{align}\label{equ:Optimal_sol_P1:eqv}
		\bm{Q}_i^{\text{opt}} = \bar{\bm{V}}_c \bar{\bm{Q}}_i^{\text{opt}}\bar{\bm{V}}_c^H + \hat{\bm{V}}_c \hat{\bm{Q}}_i^{\text{opt}}\hat{\bm{V}}_c^H,
	\end{align}
	where $\bar{\bm{Q}}_i^{\text{opt}} = \operatorname{diag}(p_{i,1}^{\text{opt}},\ldots, p_{i,r}^{\text{opt}})$ and $\hat{\bm{Q}}_i^{\text{opt}} = \operatorname{diag}(p_{i,r+1}^{\text{opt}},\ldots, p_{i,M}^{\text{opt}})$ with $\{p_{2,k}^{\text{opt}}\}_{k=1}^M$, $\{p_{3,k}^{\text{opt}}\}_{k=1}^M$, and $\{p_{4,k}^{\text{opt}}\}_{k=1}^M$ given in Propositions \ref{pro:prime_dual_relationship}, \ref{pro:prime_dual_relationship_eig}, and \ref{pro:prime_dual_relationship_det} for Scenarios 2, 3, and 4, respectively.

	It is interesting to observe from \eqref{equ:Optimal_sol_P1:eqv} that for each Scenario $i$, the transmit covariance $\bm{Q}_i^{\text{opt}}$ is separated into two parts, i.e., $\bar{\bm{V}}_c \bar{\bm{Q}}_i^{\text{opt}}\bar{\bm{V}}_c^H$ for both communication and sensing and $\hat{\bm{V}}_c \hat{\bm{Q}}_i^{\text{opt}}\hat{\bm{V}}_c^H$ for dedicated sensing only. 
	Notice that the right singular matrix $\bm{V}_c = [\bar{\bm{V}}_{c}, \hat{\bm{V}}_c]$ diagonalizes $\bm{H}_c$ into $r$ parallel subchannels. It is thus clear that $\{p_{i,k}^{\text{opt}}\}_{k=1}^r$ corresponds to the optimized power allocation over the $r$ parallel subchannels for both communication and sensing, and $\{p_{i,k}^{\text{opt}}\}_{k=r+1}^M$ corresponds to that over the other orthogonal $M-r$ dedicated sensing subchannels. 
	
	\begin{proposition}\label{lemma:Lemma_order_p}
		\emph{The optimal power allocations under Trace-CRB, MaxEig-CRB, and Det-CRB all satisfy $p_{i,1}^{\text{opt}} \geq ... \geq p_{i,r}^{\text{opt}} \geq p_{i,r+1}^{\text{opt}} = ... = p_{i,M}^{\text{opt}} > 0, \forall i \in \{2,3,4\}$.} 
	\end{proposition}
	\begin{proof}
		See Appendix \ref{Proof:lemma_power_allocation}.		
	\end{proof}
	Proposition \ref{lemma:Lemma_order_p} shows that for any Scenario $i \in \{2,3,4\},$ the power allocations $\{p_{i,k}^{\text{opt}}\}_{k=1}^r$ are monotonically increasing with respect to the subchannel gains $\{\zeta_k^2(\bm{H}_c)\}_{k=1}^r$, which is similar as the conventional water-filling power allocation in \eqref{eq:WF} for rate maximization. By contrast, the power allocations $\{p_{i,k}^{\text{opt}}\}_{k=r+1}^M$ are constant over dedicated sensing subchannels, similarly as that for CRB minimization (see Proposition \ref{Pro:three_CRB_min}). As a result, the optimal power allocation for ISAC unifies the conventional power allocations for communication only and sensing only, respectively.
	
	
	Finally, we discuss the optimal power allocation   when $P \to \infty$, where the equal power allocation is also employed over the subchannels for both communication and sensing.
	\begin{proposition}\label{pro:P_infinite}
		\emph{When $P \to \infty$, the optimal power allocations for problems (P2.3), (P3.2), and (P4.2) with Trace-CRB, Max-Eig-CRB, and Det-CRB are given by
			\begin{align}\label{equ:P_infinite}
				p_{2,k}^{\text{opt}} \rightarrow \begin{cases} \frac{1}{r} (P-\frac{(M-r)^2}{\tilde{\Gamma}_2}), & 1 \leq k \leq r \\ \frac{M-r}{\tilde{\Gamma}_2}, & r+1 \leq k \leq M \end{cases},
			\end{align}
			\begin{align}\label{equ:P_infinite_eig}
				p_{3,k}^{\text{opt}} = \begin{cases} \frac{1}{r} (P-(M-r)\tilde{\Gamma}_e), & 1 \leq k \leq r \\ \tilde{\Gamma}_e, & r+1 \leq k \leq M \end{cases},
			\end{align}
			and \begin{align}\label{equ:P_infinite_det}
				p_{4,k}^{\text{opt}} \rightarrow \begin{cases} \frac{P}{r}, & 1 \leq k \leq r \\ 0, & r+1 \leq k \leq M \end{cases}.
		\end{align}}
	\end{proposition}
	\begin{proof}
		We present the proof of $\{p_{2,k}^{\text{opt}}\}$ for Trace-CRB in Appendix \ref{Proof:Pro_P_infinite}. The other two cases can be verified similarly and thus are omitted.	
	\end{proof}

	\section{Numerical Results}\label{Section_numerical}
	
	This section provides numerical results to validate the performances of the proposed designs in the cases with point and extended targets. In the simulation, the BS-Tx, the BS-Rx, and the CU are each equipped with a ULA with half-wavelength spacing between consecutive antennas. We consider Rician fading for the communication channel, i.e., $\bm{H}_c = \sqrt{\frac{K_c}{K_c+1}} \bm{H}_c^{\text{los}} + \sqrt{\frac{1}{K_c+1}} \bm{H}_c^w$, where $\bm{H}_c^w$ is normalized to be a CSCG random matrix with zero mean and unit variance for each element, and $\bm{H}_c^{\text{los}} = \bm{a}_r^c(\theta_r^c) {\bm{a}_t^c}^T(\theta_t^c)$. Here, $\bm{a}_r^c(\theta_r^c)$ and $\bm{a}_t^c(\theta_t^c)$ denote the steering vectors at the CU receiver and the BS-Tx, $\theta_r^c$ and  $\theta_t^c$ denote the AoA at the CU and the AoD at the BS-Tx, respectively, where we set $\theta_r^c = \theta_t^c = \frac{\pi}{6}$. Furthermore, the noise power $\sigma_c^2$ at the CU and $\sigma_{s}^2$ at the BS-Rx are normalized to be unity, 
	and the number of antennas at the BS-Rx is $N_s = 12$.
	
	\subsection{Point Target Case}
	
	This subsection considers the point target case. For comparison, we consider the time switching scheme as a benchmark. In this scheme, each CPI with duration $L$ is divided into two parts with durations $L_1$ and $L_2$, in which the BS adopts transmit covariance $\bm{Q}_c^*$ and $\bm{Q}_{s,1}^*$ for rate maximization and CRB minimization, respectively, where $L_1 + L_2 = L$.
	Define $\bm{Q}_{\text{ts}} = \frac{L_1}{L} \bm{Q}_c^* + \frac{L_2}{L} \bm{Q}_{s,1}^*$. We then have the resulting estimation CRB as $\text{CRB}_1(\bm{Q}_{\text{ts}})$ and the resulting communication rate as $R_{\text{ts}} = \frac{L_1}{L} R_{\text{max}} + \frac{L_2}{L} R_{1,S}$. By adjusting $L_1$ and $L_2$, we can get different boundary points to balance the tradeoff between the estimation CRB and the communication rate.  
	\begin{figure}[htb]
		\centering
		\setlength{\abovecaptionskip}{+4mm}
		\setlength{\belowcaptionskip}{+1mm}
		\subfigure[C-R region ]{ \label{fig:point_region_tlessthanM}
			\includegraphics[width=3.1in]{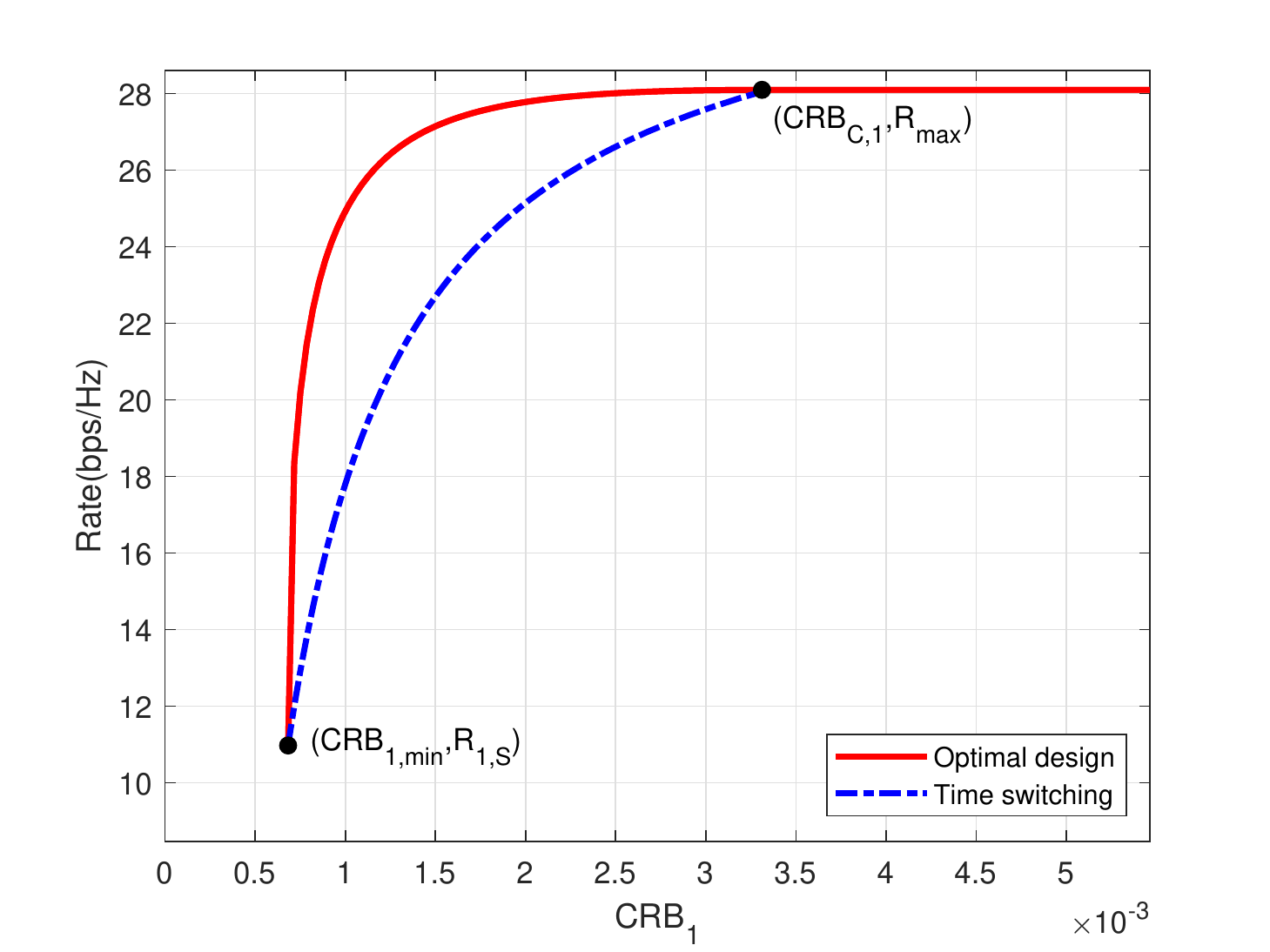}}
		\subfigure[Rate versus SNR with $\Gamma_1 = 0.01$]{ \label{fig:rate_vs_snr_point}
			\includegraphics[width=3.1in]{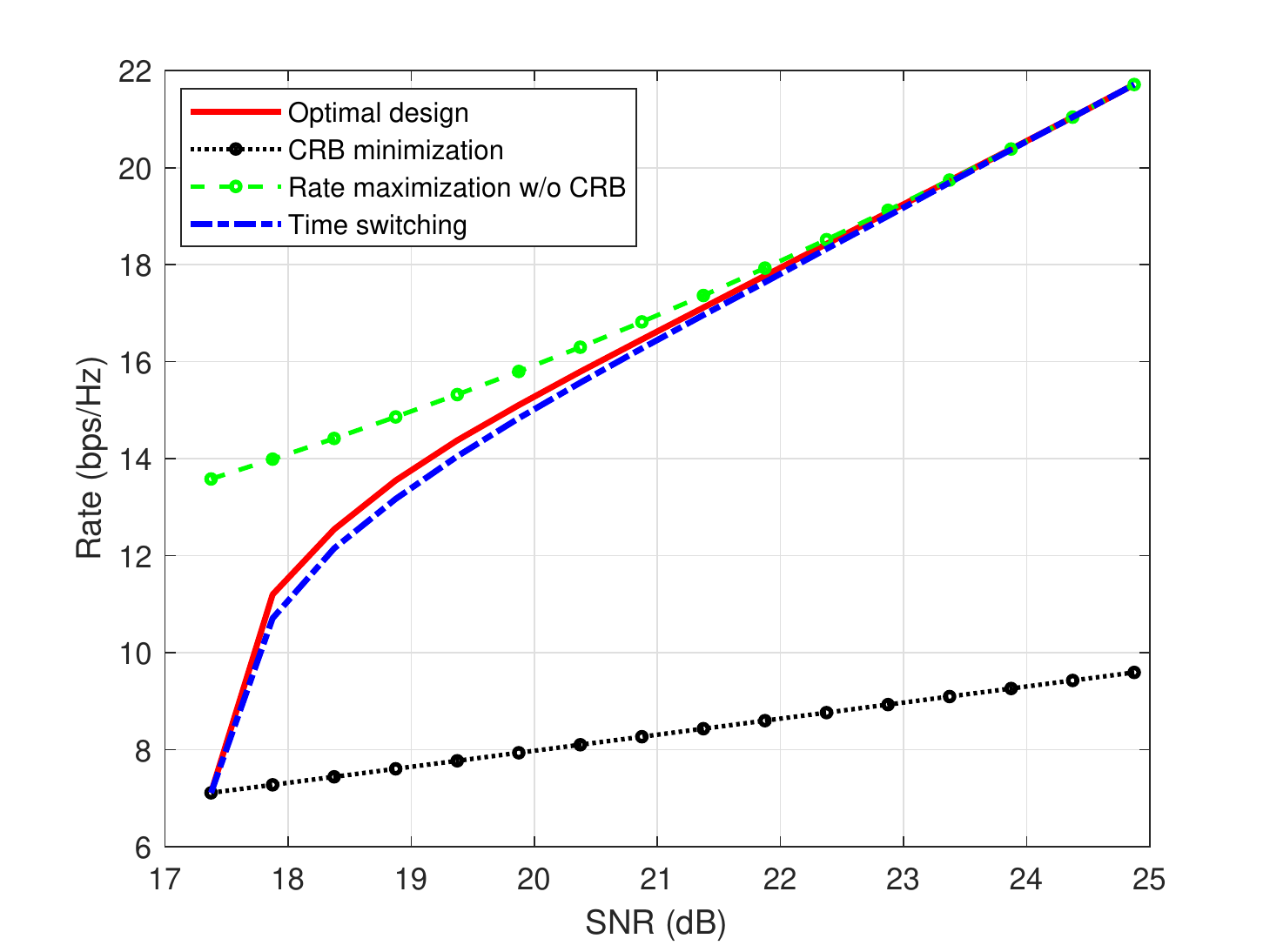}}
		\caption{Point target case with $M = 8 < N_s$, $r = N_c = 6$, and the AoA/AoD of the point target $\theta = -0.2803 \pi$. }
		\label{fig:point_sim}
	\end{figure}
	
	In the simulation, we set the number of antennas at the BS-Tx $M = 8$, the number of antennas at the CU $N_c = 6$, and the angle of the point target $\theta = -0.2803 \pi$. We set the coefficient $\alpha = 10^{-3}$, $P = 800$, and the Rician factor as $K_c = 100$. Fig. \ref{fig:point_region_tlessthanM} shows the obtained C-R-region boundary. It is observed that the optimal design outperforms the time switching scheme. Fig. \ref{fig:rate_vs_snr_point} shows the rate versus the signal-to-noise ratio (SNR) (i.e., $\frac{P}{\sigma_c^2}$) with the CRB threshold being $\text{CRB}_1 \leq \Gamma_1 = 0.01$, based on which problem (P1) is only feasible when the SNR becomes higher than 17.3 dB.
	We also consider the CRB minimization design and the rate maximization design as the rate performance lower bound and upper bound, respectively. When the SNR is close to 17.3 dB, the optimal design and the time switching scheme are observed to perform similar as the rate lower bound by the CRB minimization design. When the SNR becomes high, the optimal design and the time switching scheme are observed to perform close to the rate upper bound by rate maximization. This is due to the fact that the CRB constraint may become inactive in this case, and thus both schemes become identical to the rate maximization.

	\subsection{Extended Target Case}
	This subsection evaluates the performance of our proposed designs in the extended target case, as compared to the following benchmark schemes.
	\begin{itemize}
		\item {\bf Time switching}: For each Scenario $i \in \{2,3,4\}$, the BS time switches between the two transmit covariances $\bm{Q}_c^*$ and $\bm{Q}_{s,i}^*$ for rate maximization and CRB minimization, respectively, similarly as that in the previous subsection. Notice that this scheme is only applicable when $\bm{Q}_c^*$ is of full rank, since otherwise $\text{CRB}_{C,i} = \text{CRB}(\bm{Q}_c^*) \rightarrow \infty, \forall i \in \{2,3,4\}$ follows according to Remark \ref{remark:finite_SCRB}. 
		\item {\bf Power splitting with equal power allocation (EP)}: The BS sets the transmit covariance as $\bm{Q}^{\text{EP}} = \bm{V}_c \tilde{\bm{Q}}^{\text{EP}}\bm{V}_c^H$, in which $\tilde{\bm{Q}}^{\text{EP}} = \operatorname{diag}(p_1^{\text{EP}},\ldots, p_M^{\text{EP}})$ denotes the power allocation. The BS splits the transmit power $P$ into two parts, $\beta P$ for the first $r$ ISAC subchannels and $(1-\beta)P$ for the $M-r$ dedicated sensing subchannels, with $0\le \beta\le 1$ denoting the power splitting factor that is a variable to be optimized. Following the equal power allocation, we have $p_1^{\text{EP}} = \ldots = p_r^{\text{EP}} =  \frac{\beta P}{r}$ and $p_{r+1}^{\text{EP}} = \ldots = p_M^{\text{EP}} =  \frac{(1-\beta) P}{M-r}$. Notice that if $r = M$, then we set $\beta = 1$.
		\item {\bf Power splitting with strongest eigenmode transmission (SEM)}: The BS sets $\bm{Q}^{\text{SEM}} = \bm{V}_c \tilde{\bm{Q}}^{\text{SEM}}\bm{V}_c^H$, in which $\tilde{\bm{Q}}^{\text{SEM}} = \operatorname{diag}(p_1^{\text{SEM}},\ldots, p_M^{\text{SEM}})$. The BS splits the transmit power $P$ into two parts, $\beta P$ for the the dominant ISAC subchannel and $(1-\beta)P$ for the remaining $M-1$ subchannels, i.e., $p_1^{\text{SEM}} = \beta P$ and $p_{2}^{\text{SEM}} = \ldots = p_M^{\text{EP}} =  \frac{(1-\beta) P}{M-1}$, with $0 \leq \beta \leq 1$ to be optimized.
	\end{itemize}
	
	
	First, we consider the scenario where $M = 8$, $N_c = 6$, $K_c=100$, and $P = 800$. In this case, we have $r < M$, such that $\bm{Q}^*_c$ is rank-deficient and $\text{CRB}_{C,i} \rightarrow \infty, \forall i \in \{2,3,4\}$, for which the time switching is not applicable. Figs. \ref{fig:region_tlessthanM}, \ref{fig:extended_region_tlessthanM_eig}, and \ref{fig:extended_region_tlessthanM_det} show the resultant C-R regions with Trace-CRB, MaxEig-CRB, and Det-CRB, respectively.
	It is observed that for all the three sensing performance measures, the C-R-region boundary curve achieved by the optimal design outperforms those by the power splitting designs with equal power allocation and strongest eigenmode transmission. It is also observed that when the CRB is low (or the CRB constraint becomes tight), the curves obtained from the two suboptimal designs approach the optimal C-R-region boundary curve.
	Furthermore, when the CRB is high (or the CRB constraint becomes relaxed), the C-R-region boundary curve is observed to approach the capacity without sensing (i.e., $R_{\text{max}}$). This is consistent with the result in Remark \ref{remark:finite_SCRB}. 
	
	\begin{figure}[htb]
		\centering
		\setlength{\abovecaptionskip}{+4mm}
		\setlength{\belowcaptionskip}{+1mm}
		\subfigure[C-R region with Trace-CRB.]{ \label{fig:region_tlessthanM}
			\includegraphics[width=2.0in]{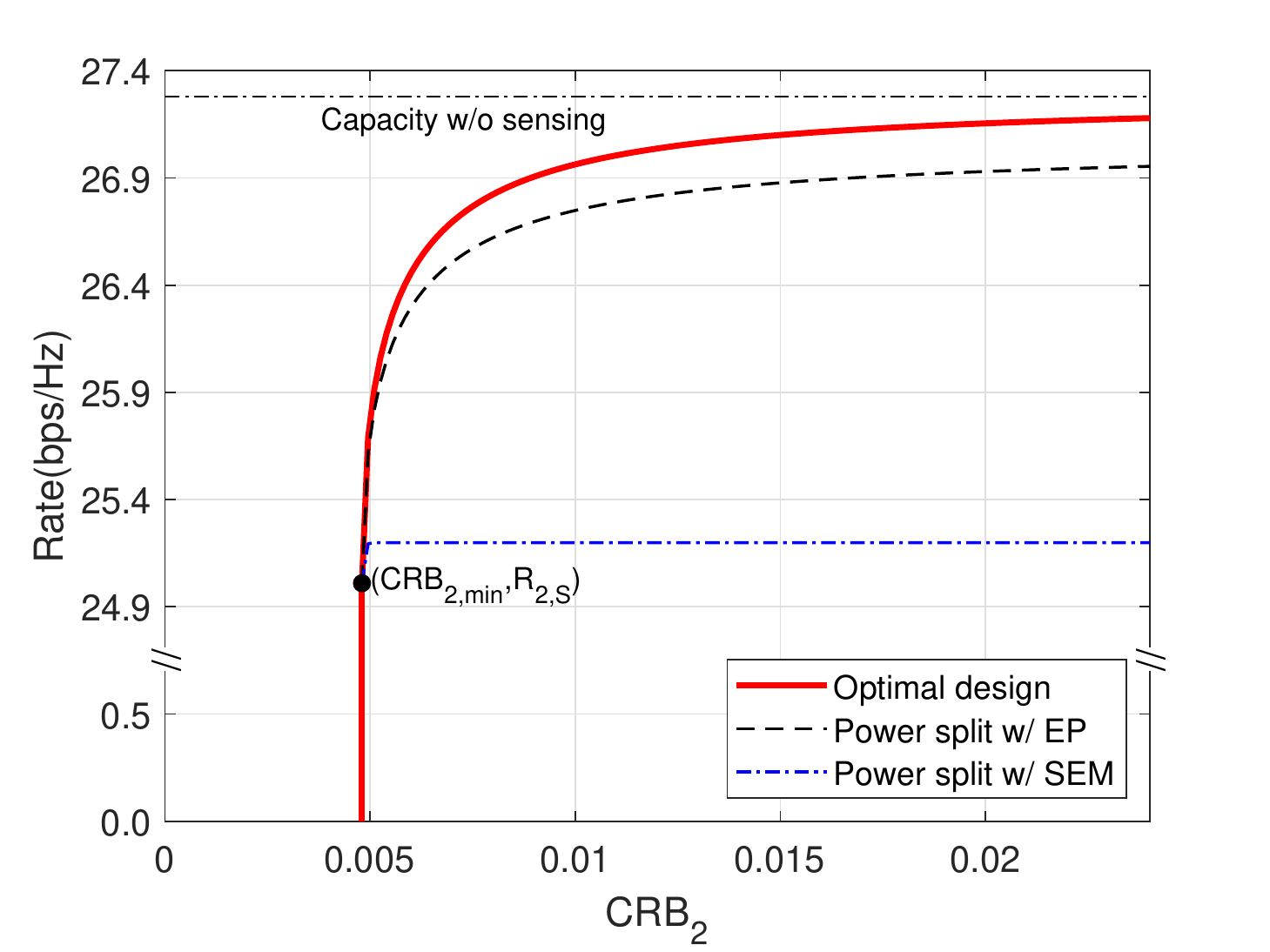}}
		\subfigure[C-R region with MaxEig-CRB.]{ \label{fig:extended_region_tlessthanM_eig}
			\includegraphics[width=2.0in]{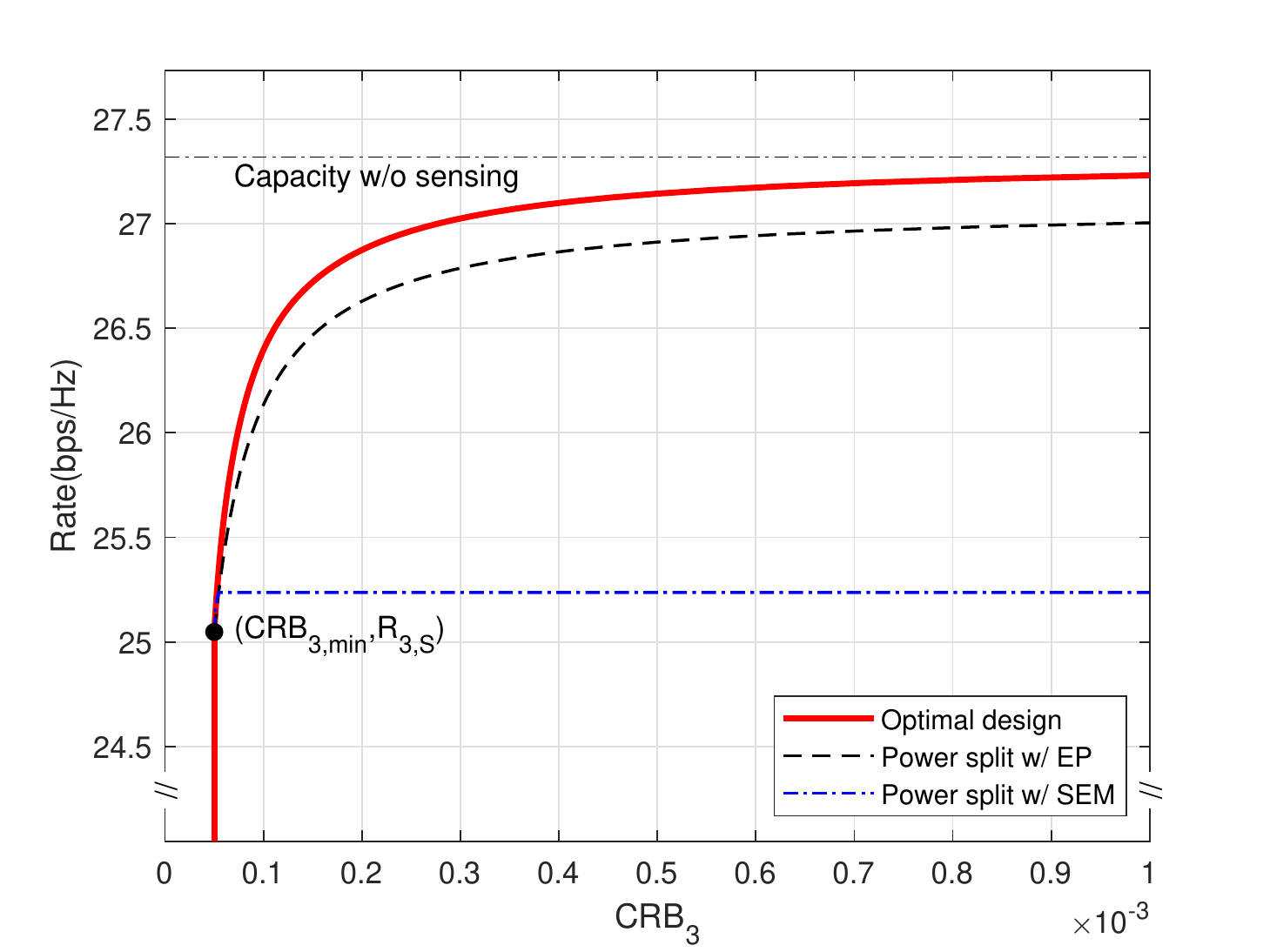}}
		\subfigure[C-R region with Det-CRB.]{ \label{fig:extended_region_tlessthanM_det}
			\includegraphics[width=2.0in]{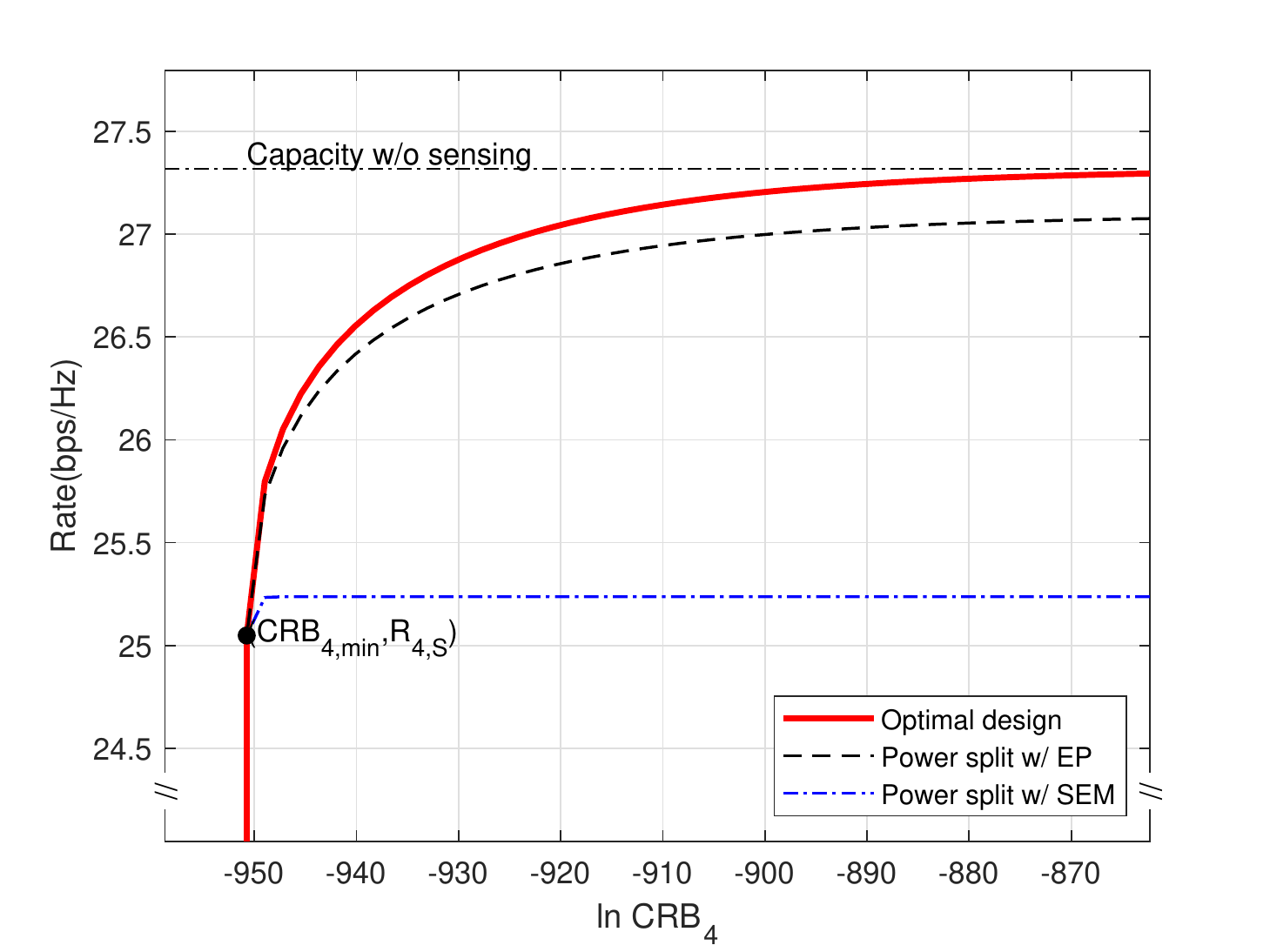}}
		\caption{The C-R region in the case with $M = 8$ and $r = N_c = 6$. }
		\label{fig:CR_region_rank_def}
	\end{figure}
	
	Figs. \ref{fig:power_allocation_deficient}, \ref{fig:power_allocation_deficient_eig}, and \ref{fig:power_allocation_deficient_det} show the optimal power allocation by the optimal designs with Trace-CRB, MaxEig-CRB, and Det-CRB,
	where $\Gamma_2 = 0.0152$, $\Gamma_3 = 8 \times 10^{-4}$, and $\ln \Gamma_4 = -900$, respectively. The water-filling power allocation for rate maximization and the equal power allocation for CRB minimization are considered for comparison. It is observed that for each scenario, the proposed optimal power allocations over the first six ISAC subchannels are monotonically non-increasing, which are higher than the constant power allocated to the last two sensing subchannels. This is consistent with Proposition \ref{lemma:Lemma_order_p}. It is also observed that the proposed optimal power allocations over the first five communication subchannels are lower than the corresponding water-filling power allocations, as more power should be allocated to other subchannels for sensing. By contrast, the proposed optimal power allocations over the last two subchannels are observed to be higher than the corresponding water-filling power allocations (which are zero), in order to meet the sensing requirement. 
	
	\begin{figure}[htb]
		\centering
		\setlength{\abovecaptionskip}{+4mm}
		\setlength{\belowcaptionskip}{+1mm}
		\subfigure[Trace-CRB with $\Gamma_2 = 0.0152$.]{ \label{fig:power_allocation_deficient}
			\includegraphics[width=2.0in]{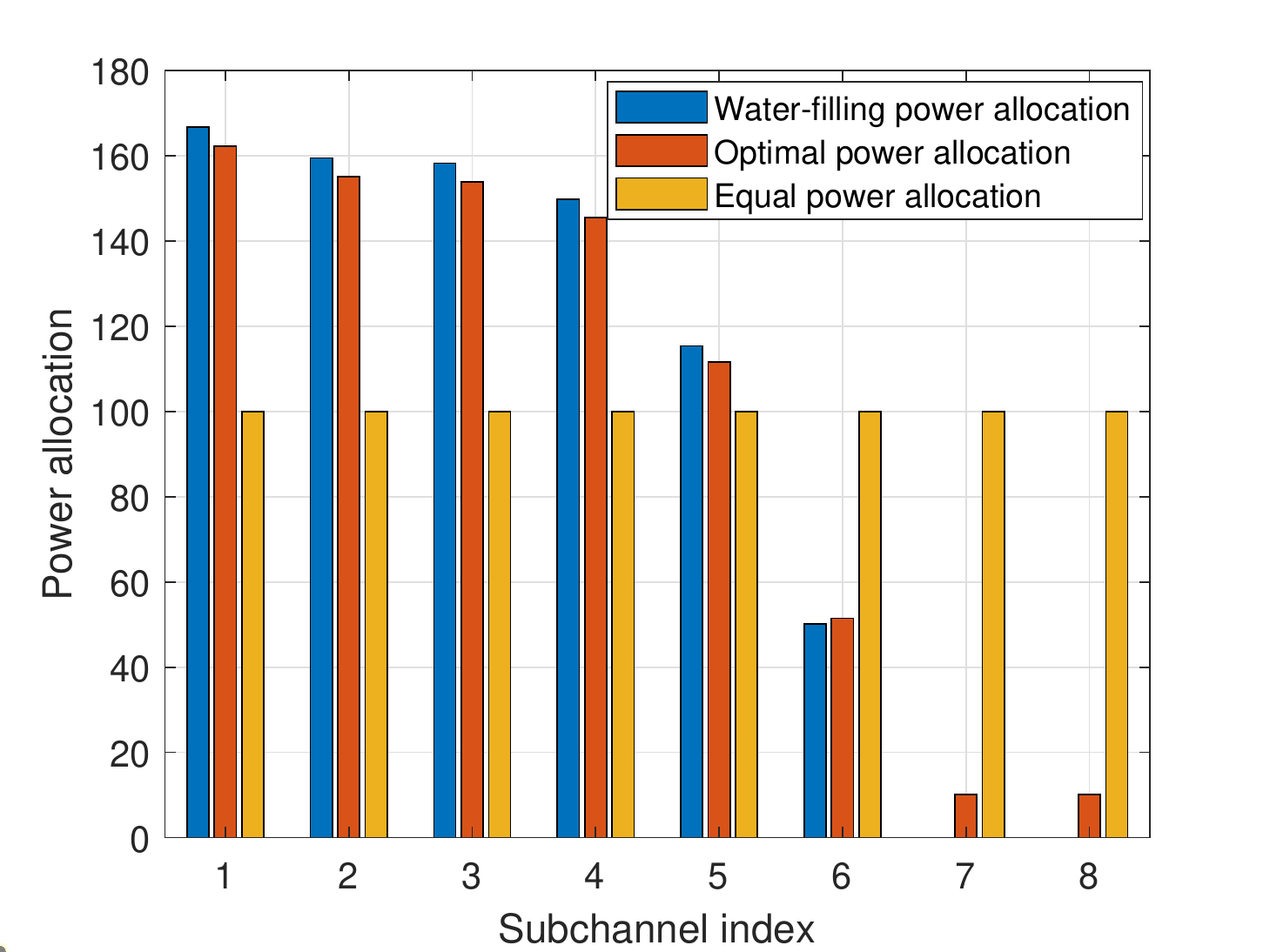}}
		\subfigure[MaxEig-CRB with $\Gamma_3 = 8 \times 10^{-4}$.]{ \label{fig:power_allocation_deficient_eig}
			\includegraphics[width=2.0in]{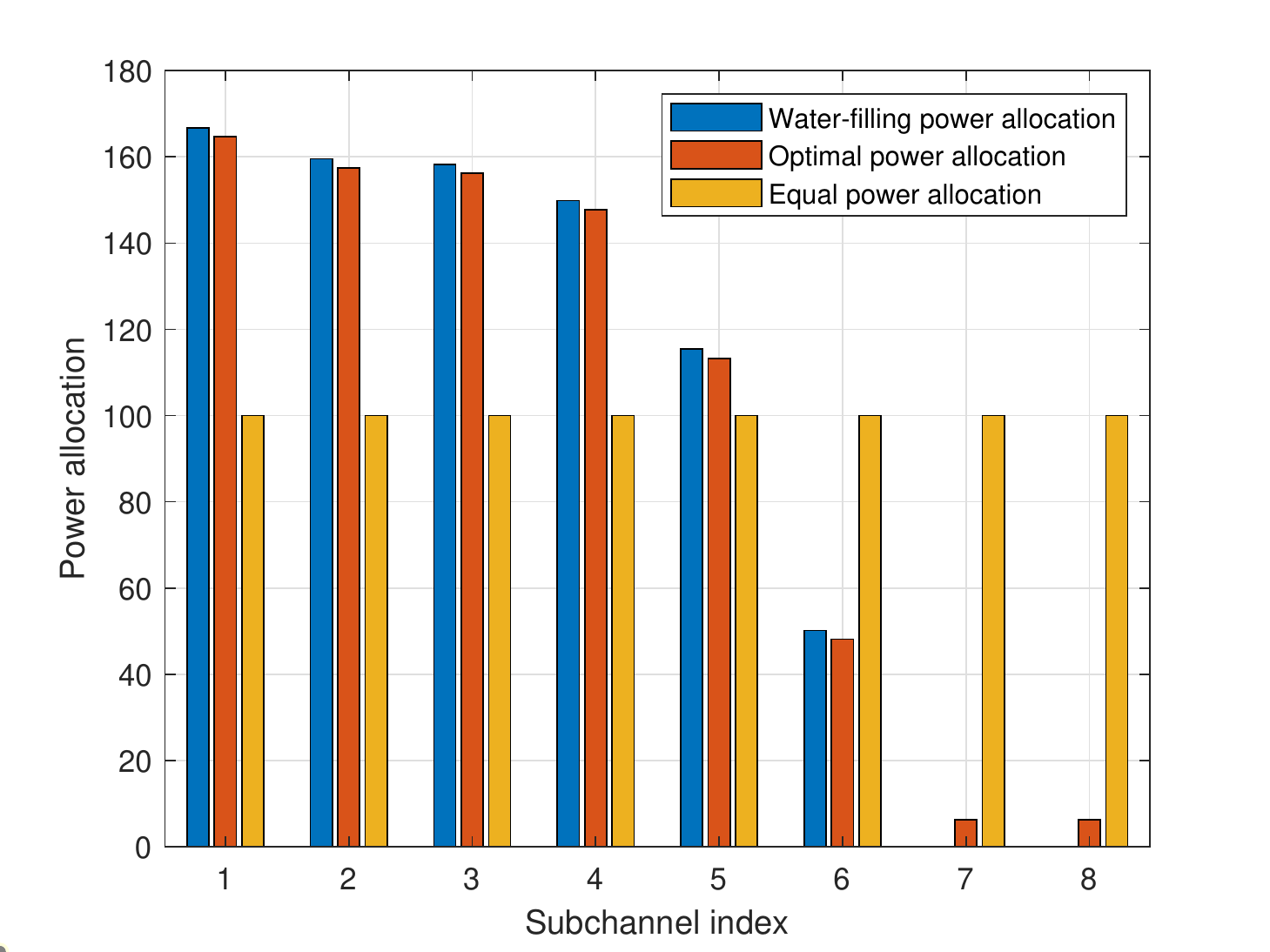}}
		\subfigure[Det-CRB with $\ln \Gamma_4 = -900$.]{ \label{fig:power_allocation_deficient_det}
			\includegraphics[width=2.0in]{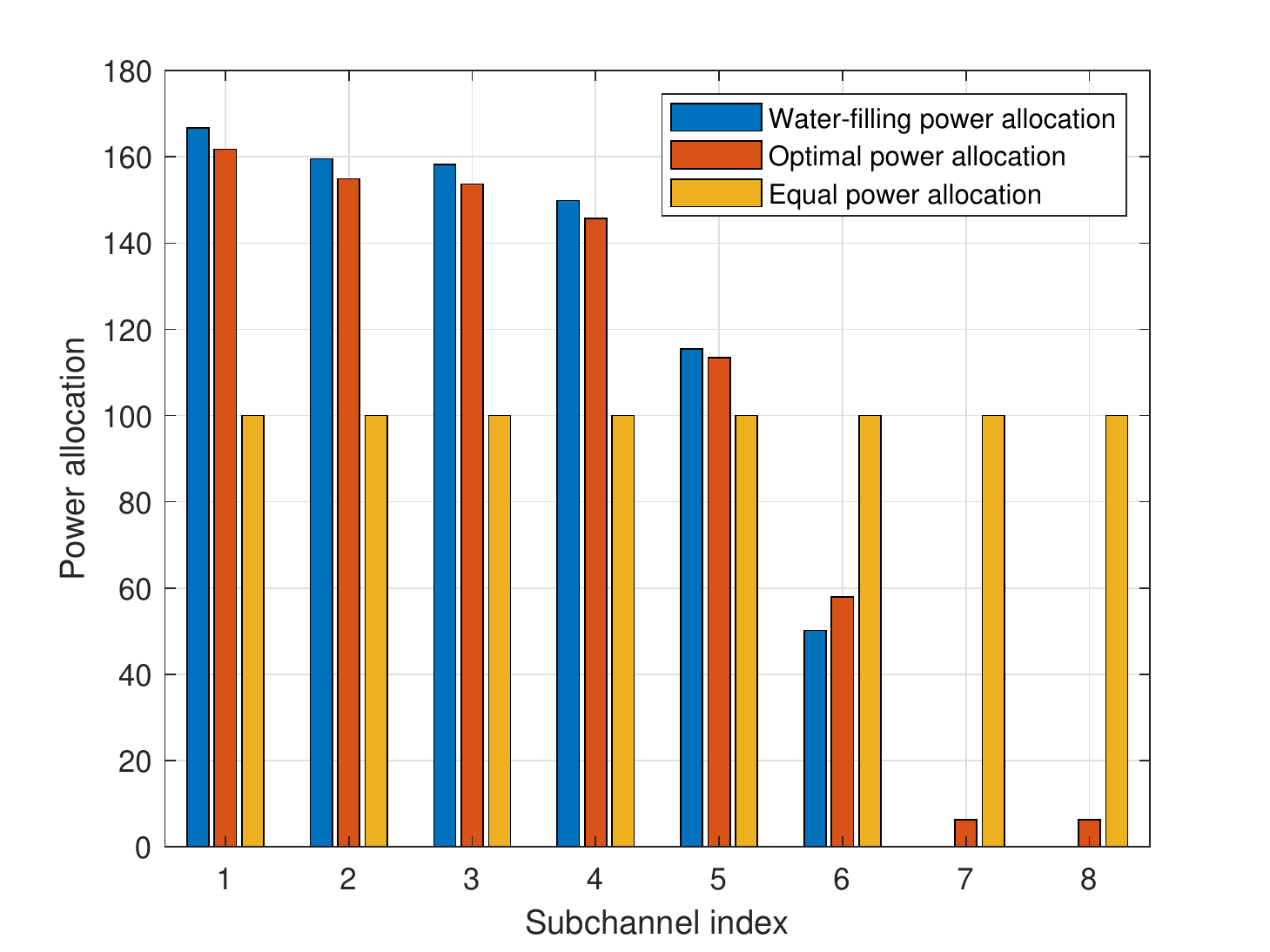}}
		\caption{The power allocation by the proposed optimal designs with $M = 8$ and $r = N_c = 6$. }
		\label{fig:Power_allocation_rank_def}
	\end{figure}

	\begin{figure}[htb]
		\centering
		\setlength{\abovecaptionskip}{3mm}
		\setlength{\belowcaptionskip}{3mm}
		\includegraphics[width=2.7in]{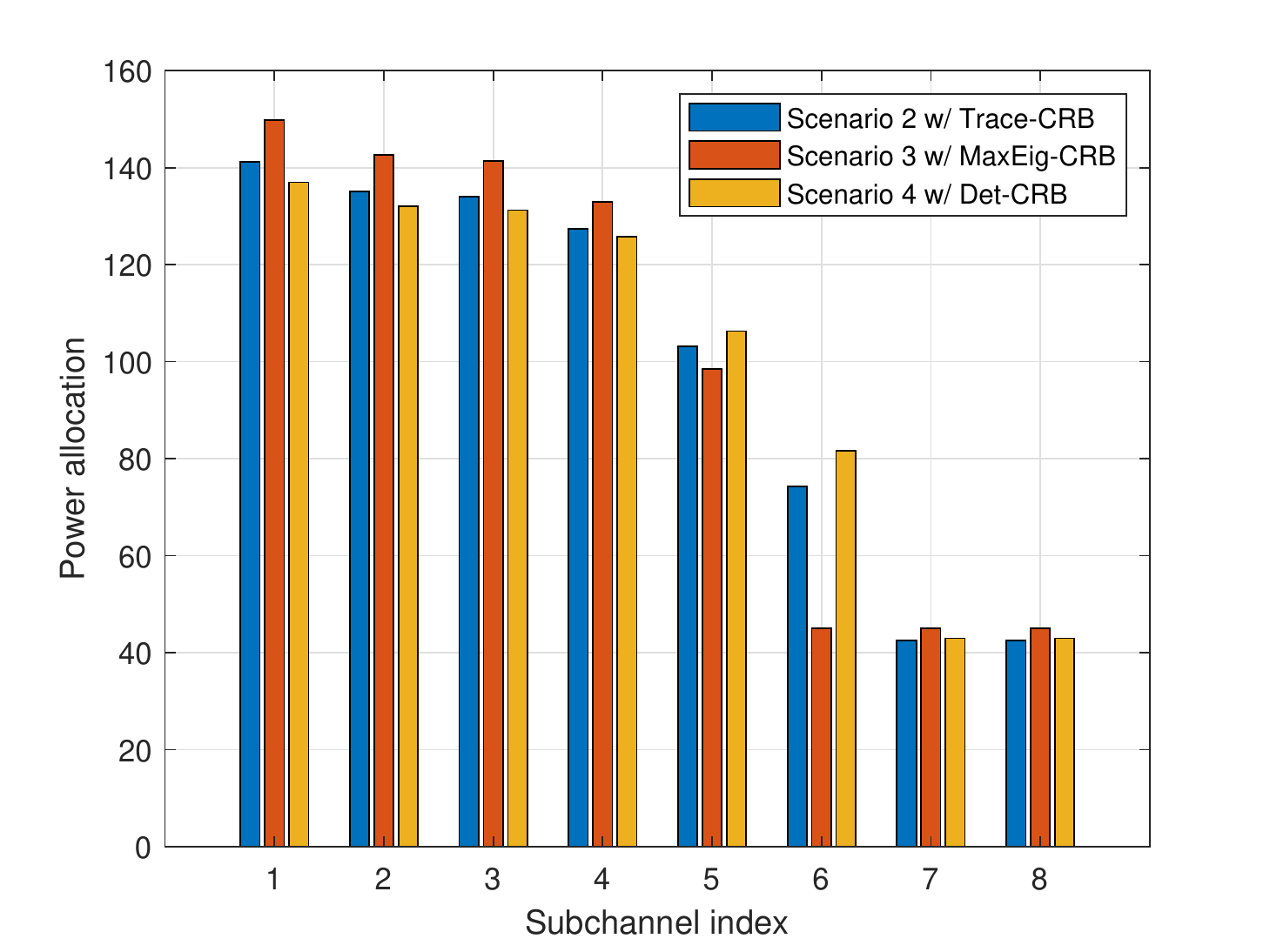}
		\caption{The power allocation under the three scenarios when the rate $R = 26.5$ bps/Hz. } \label{fig:Three_criteria_power_comparison}
	\end{figure}

	Fig. \ref{fig:Three_criteria_power_comparison} compares the optimal power allocation by our proposed designs under the three sensing performance measures, where the resultant data rate is set as $R = 26.5$ bps/Hz for fair comparison. 
	In particular, look at the transmit power allocated to subchannels 7-8 for dedicated sensing. It is observed that the allocated power to subchannels 7-8 in Scenario 2 is lowest among the three scenarios, as allocating more power to ISAC subchannels 1-6 is also beneficial in minimizing the Trace-CRB. Next, it is observed that  the allocated power to subchannels 7-8 in Scenario 3 is highest among the three scenarios, as the design based on MaxEig-CRB ensures that the upper bound of the worst-case CRB is minimized. Finally, the allocated power to subchannels 7-8 in Scenario 4 is observed to lie between the above two scenarios, as the Det-CRB metric ensures the CRB minimization among different elements in a proportional fair manner.


	\begin{figure}[htb]
		\centering
		\setlength{\abovecaptionskip}{+4mm}
		\setlength{\belowcaptionskip}{+1mm}
		\subfigure[Rate versus SNR with $\Gamma_2 = 0.1$.]{ \label{fig:rate_SNR} 
			\includegraphics[width=2.0in]{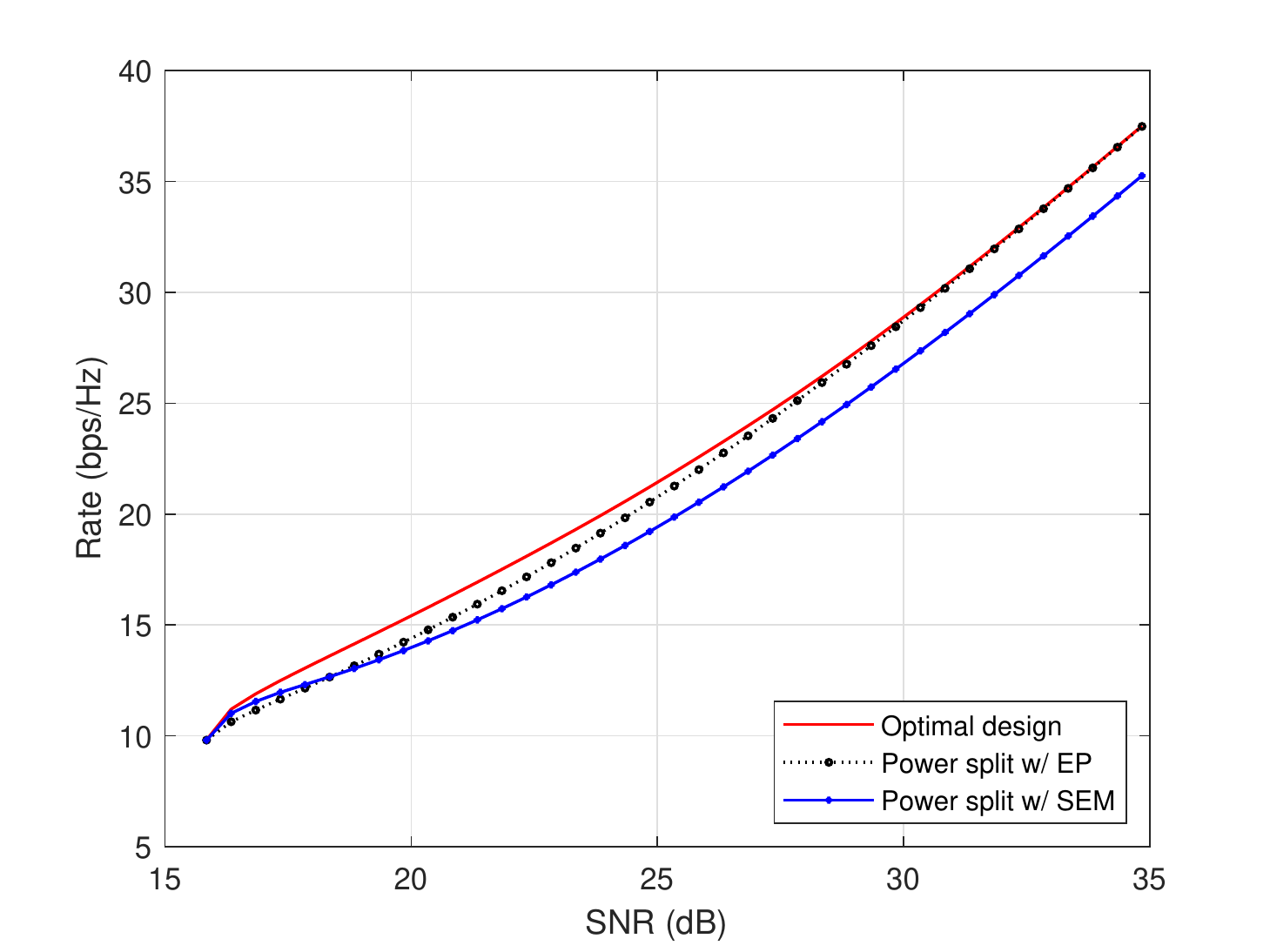}}
		\subfigure[Rate versus SNR with $\Gamma_3 = 8 \times 10^{-4}$.]{ \label{fig:rate_SNR_eig} 
			\includegraphics[width=2.0in]{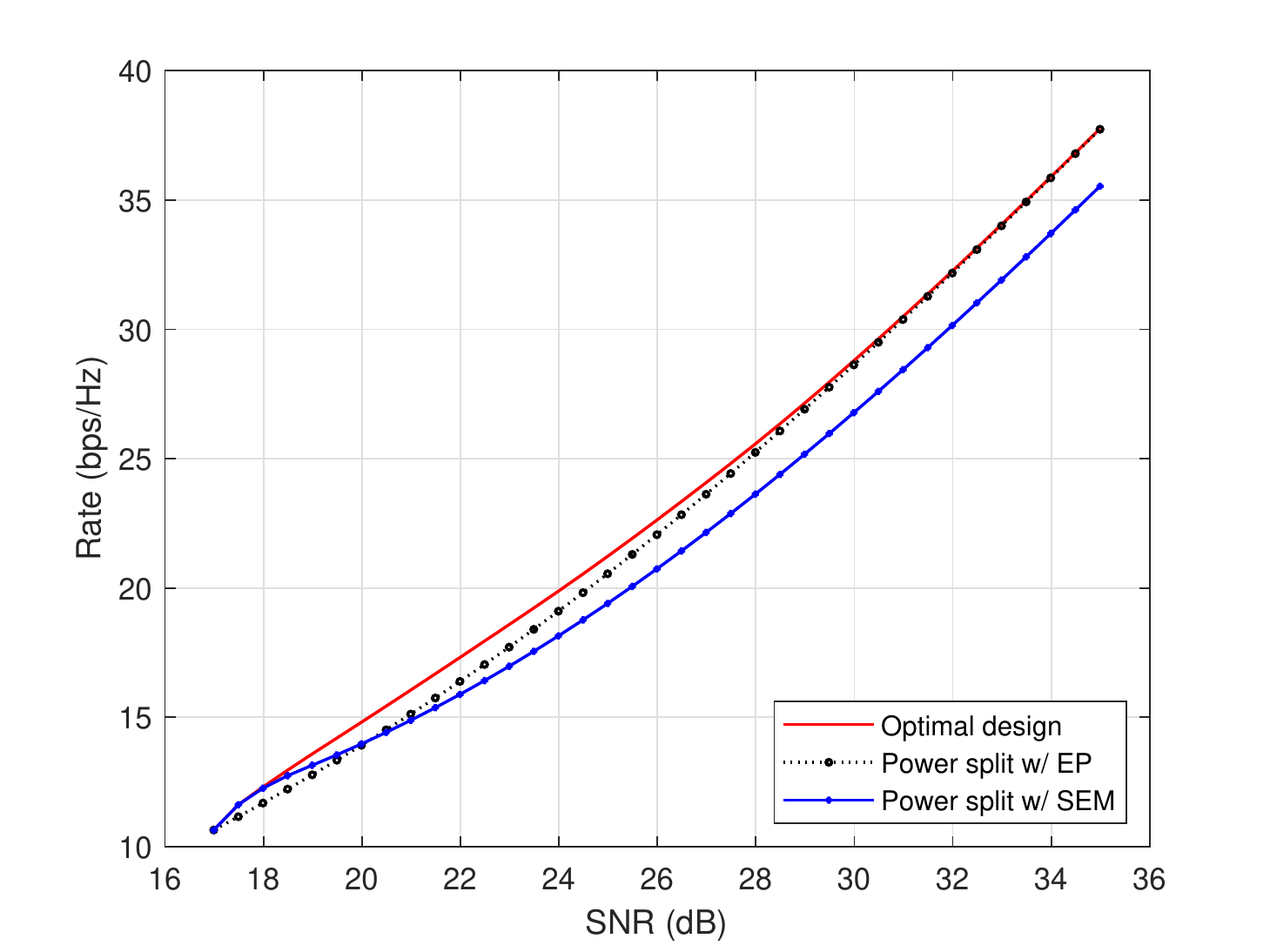}}
		\subfigure[Rate versus SNR with $\ln \Gamma_4 = -900$.]{ \label{fig:rate_SNR_det}
			\includegraphics[width=2.0in]{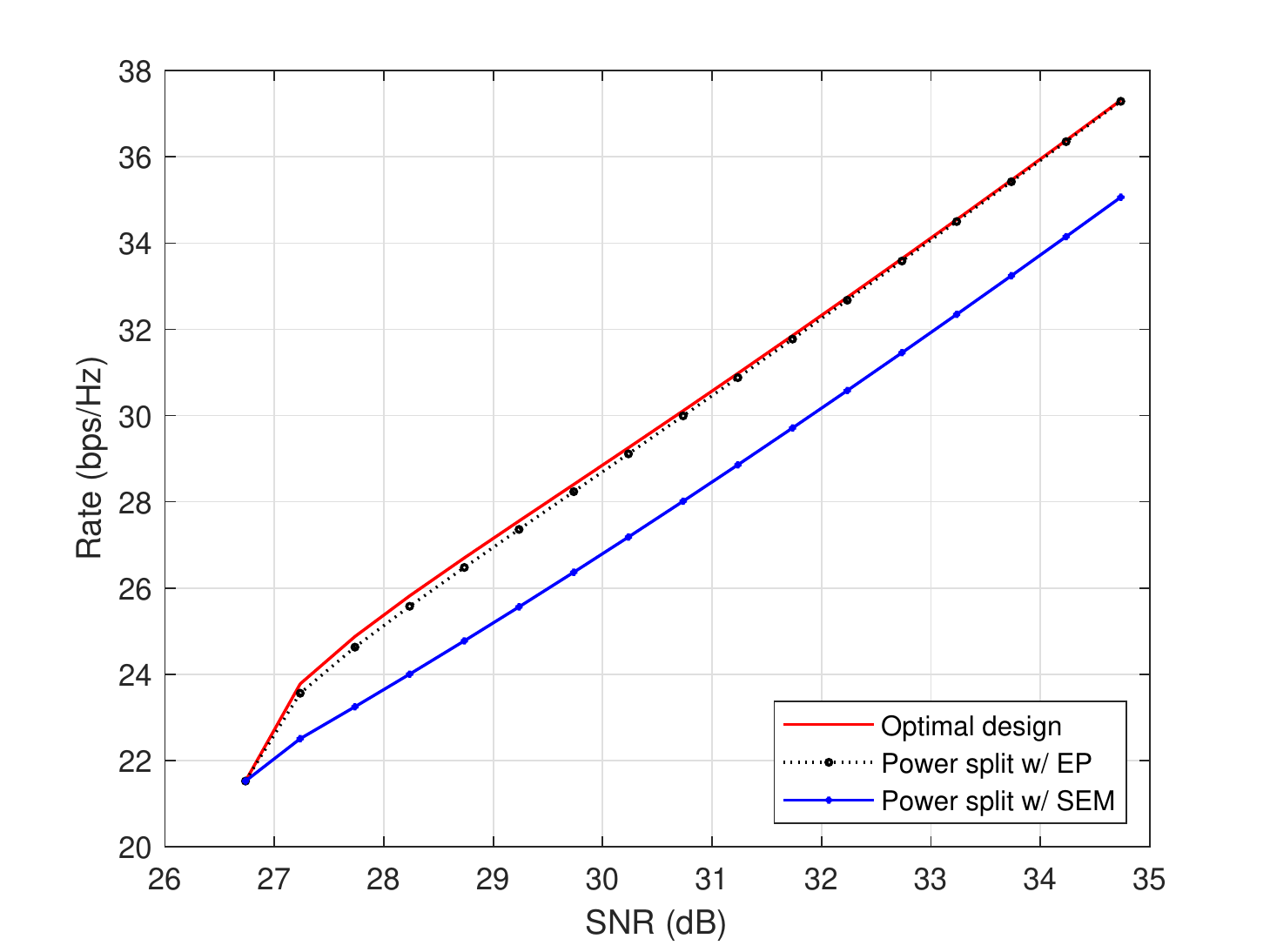}}
		\caption{Rate versus SNR in the case with $r = M = N_c = 6$. }
		\label{fig:Rate_vs_SNR}
	\end{figure}
	
	Figs. \ref{fig:rate_SNR}, \ref{fig:rate_SNR_eig}, and \ref{fig:rate_SNR_det} show the rate versus SNR in Scenarios 2-4 with Trace-CRB, MaxEig-CRB, Det-CRB, where $\Gamma_2 = 0.1$, $\Gamma_3 = 8 \times 10^{-4}$, and $\ln \Gamma_4 = -900$, respectively. It is observed that for each scenario, the optimal design performs best over the whole SNR regime. In the high SNR regime, the rate achieved by the power splitting with equal power allocation is observed to approach that by the optimal design. This can be explained by Proposition \ref{pro:P_infinite}. In the low SNR, the power splitting with strongest eigenmode transmission is observed to approach the optimal design.

	\begin{figure}[htb]
		\centering
		\setlength{\abovecaptionskip}{+4mm}
		\setlength{\belowcaptionskip}{+1mm}
		\subfigure[C-R region with Trace-CRB.]{ \label{fig:C-R-region}
			\includegraphics[width=2.0in]{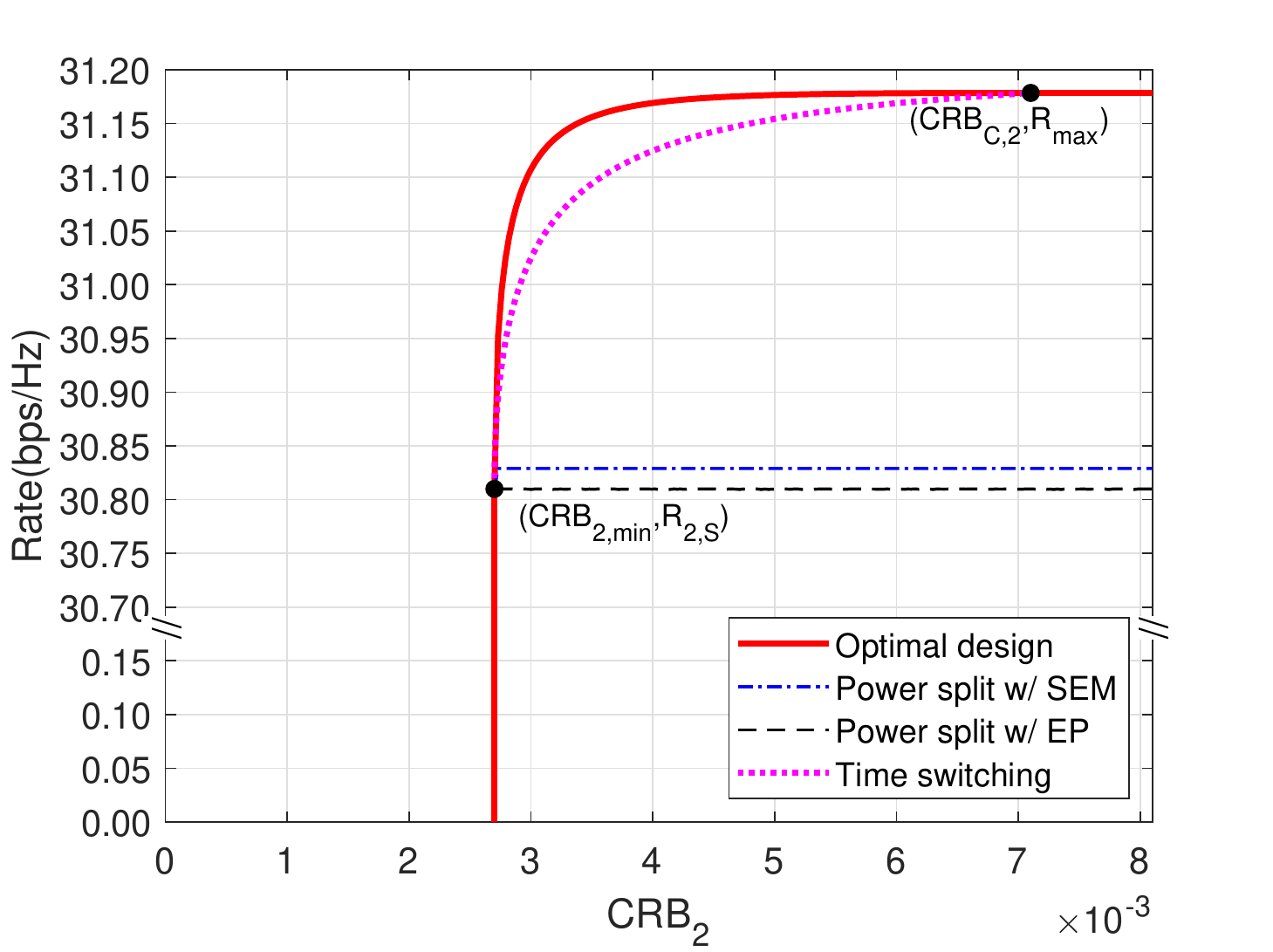}}
		\subfigure[C-R region with MaxEig-CRB.]{ \label{fig:extended_region_teqM_eig}
			\includegraphics[width=2.0in]{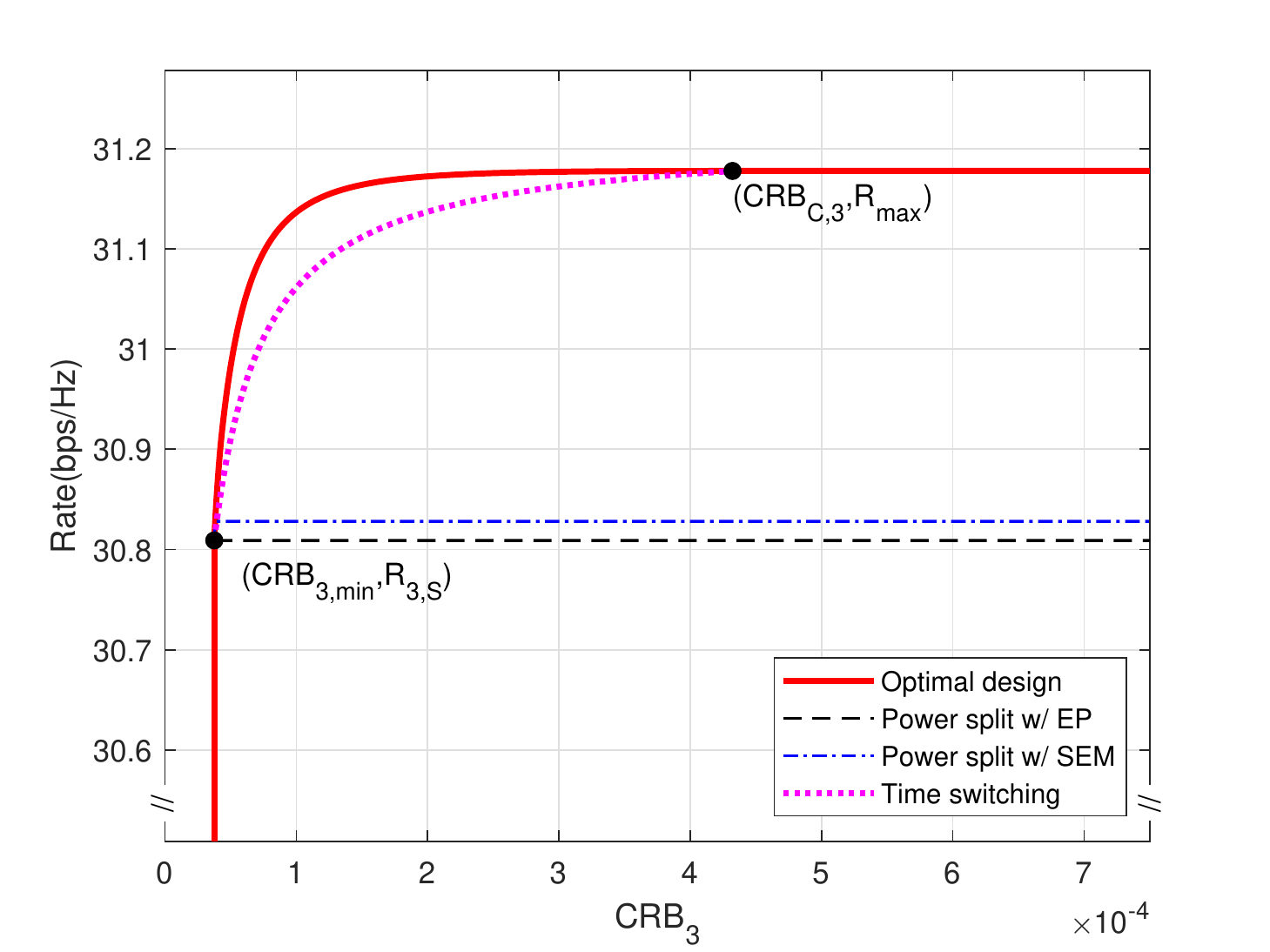}}
		\subfigure[C-R region with Det-CRB.]{ \label{fig:extended_region_teqM_det}
			\includegraphics[width=2.0in]{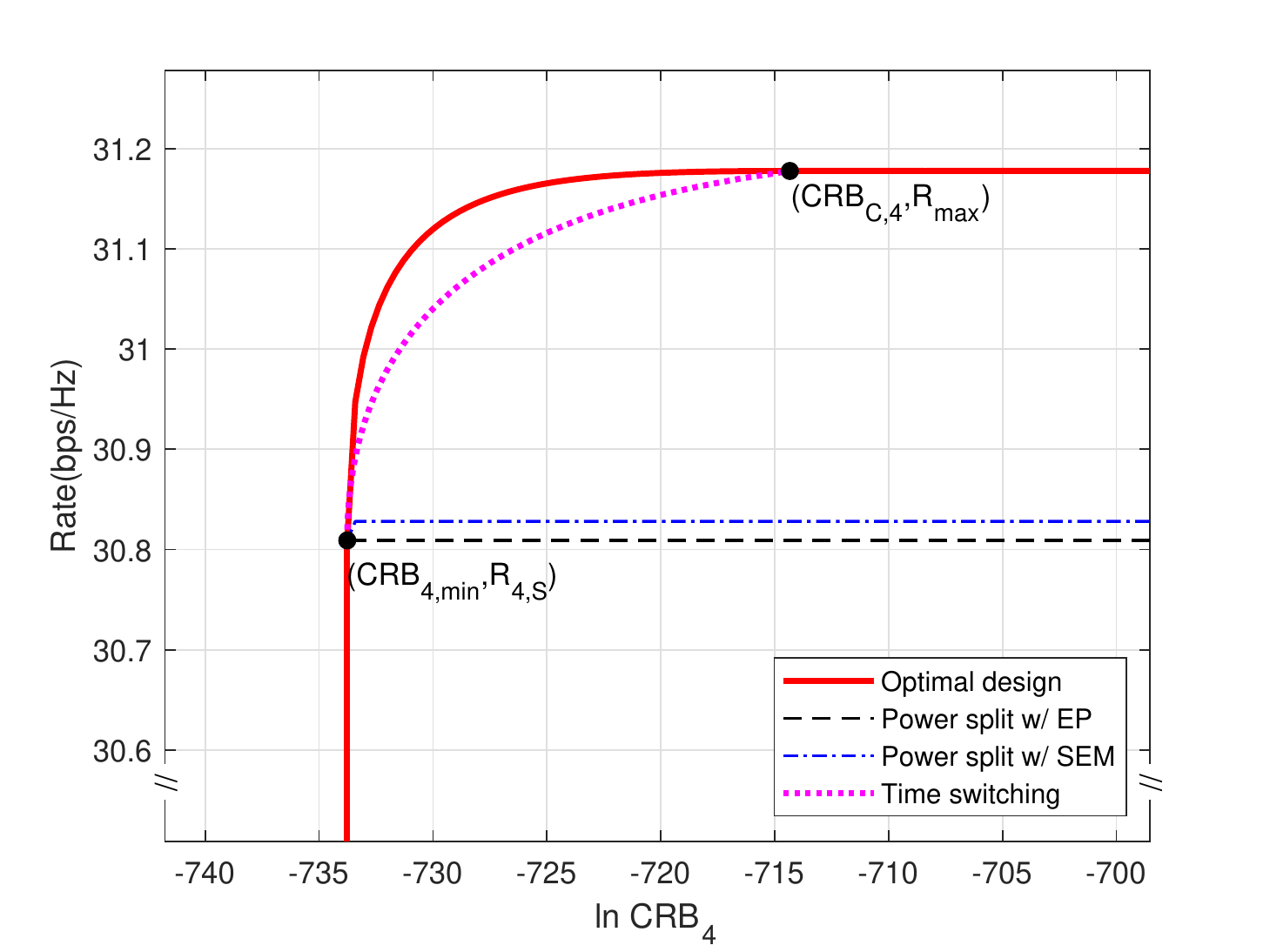}}
		\caption{C-R region in the case with $r = M = N_c = 6$. }
		\label{fig:CR_region_full_rank}
	\end{figure}
	
	Finally, we consider that $M = N_c = 6$, $K_c=20$, and $P = 800$. In this case, we have $r = 6$, and $\bm{Q}_c^*$ is of full rank (as $P > P_0$ in  Remark \ref{remark:finite_SCRB}) such that $\text{CRB}_{C,i}$ is finite for any $i \in \{2,3,4\}$. Figs. \ref{fig:C-R-region}, \ref{fig:extended_region_teqM_eig}, and \ref{fig:extended_region_teqM_det} show the resultant C-R regions achieved by the optimal designs for Scenarios 2-4. It is observed that $(\text{CRB}_{C,i}, R_{\text{max}})$ exists and the C-R-region boundary achieved by the optimal design outperforms other three benchmark schemes. Among the three benchmark schemes, the performance of the time switching design is the best.

	\section{Conclusion}
	
	This paper investigated the fundamental performance tradeoff between the estimation CRB and the communication data rate in a point-to-point MIMO ISAC system, by considering both point and extended target models. We characterized the complete Pareto boundary of the resultant C-R regions, by formulating new CRB-constrained MIMO rate maximization problems and finding their semi-closed-form optimal transmit covariance solutions. Numerical results showed that the C-R-region boundary achieved by the optimal design significantly outperforms other benchmark schemes. The fundamental C-R tradeoff limits revealed in this paper are expected to provide references and design insights on practical ISAC systems. 
	

	\appendix
	
	
	\subsection{Proof of Proposition \ref{Pro:three_CRB_min}} \label{three_CRB_min_proof}
	
	For Scenario $i=2$, we can find that the optimal solution to $\min _{\bm{Q}\succeq \bm{0}, \operatorname{tr}(\bm{Q}) \leq P}  \text{ } \frac{\sigma_{s}^2 N_s}{L} \operatorname{tr}(\bm{Q}^{-1})$ is  $\bm{Q}_{s,2}^{*} = \frac{P}{M} \bm{I}_M$ \cite{liu2021cramer} by checking the KKT conditions.
	
	For Scenario $i=3$, minimizing the maximium eigenvalue of $\bm{Q}^{-1}$ is equivalent to maximizing the minimum eigenvalue of $\bm{Q}$. As a result, $\min _{\bm{Q}\succeq \bm{0}, \operatorname{tr}(\bm{Q}) \leq P}  \text{ } \frac{\sigma_s^2}{L} \lambda_{\text{max}} (\bm{Q}^{-1})$ is equivalent to
	\begin{align}\label{equ:P_CRB_eigen_opt_eq}
		\max _{\bm{Q}\succeq \bm{0}, t} & \text{ } t, \quad  \text { s.t. }  \bm{Q} \geq t \bm{I}, \quad \operatorname{tr}(\bm{Q}) \leq P.
	\end{align}
	The optimal solution $\bm{Q}_{s,3}^{*}$ to (\ref{equ:P_CRB_eigen_opt_eq}) is shown to be a diagonal matrix by applying the Schur-Horn Theorem \cite{horn1954doubly}.
	Let $\bm{Q}_{s,3}^{*} = \operatorname{diag}(p_1,...,p_M)$. Problem (\ref{equ:P_CRB_eigen_opt_eq}) becomes
	\begin{align} \label{equ:P_CRB_eigen_opt_sim}
		\max _{\{p_k\}, t}  \text{ } t, \quad \text {s.t. }  p_k \geq t, \forall k \in \{1,2,...,M\}, \quad \sum_{i=1}^M p_k \leq P.
	\end{align}
	The optimal solution to problem (\ref{equ:P_CRB_eigen_opt_sim}) is $p_k = \frac{P}{M}, \forall k$. Accordingly, we have $\bm{Q}_{s,3}^{*} = \frac{P}{M} \bm{I}_M$.
	
	For Scenario $i=4$, based on the Hadamard inequality \cite{horn2012matrix}, it is clear that the optimal solution $\bm{Q}_{s,4}^*$ to problem $\min _{\bm{Q}\succeq \bm{0}, \operatorname{tr}(\bm{Q}) \leq P} \text{ } (\frac{\sigma_s^2}{L})^{M N_s} \operatorname{det}(\bm{Q}^{-1})^{N_s}$ must be a diagonal matrix. It can be shown that all its diagonal elements should be $\frac{P}{M}$, and thus we have $\bm{Q}_{s,4}^{*} = \frac{P}{M} \bm{I}_M$. 
	
	
	\subsection{Proof of Proposition \ref{Pro:diagonal_optimal}} \label{Trace_diag_optimal_proof}
	
	First, it is evident that $\tilde{\bm{Q}}_2 \succ \bm{0}$ follows in order for the maximum CRB constraint in (P2.2) to hold. Next, suppose that the optimal solution $\tilde{\bm{Q}}_2^*$ is not diagonal, and we construct an alternative solution $\tilde{\bm{Q}}_2^{**} = \tilde{\bm{Q}}_2^* \circ \mv{I}$, which is a diagonal matrix whose diagonal elements are identical to $\tilde{\bm{Q}}_2^*$. Then, we have $\log_2 \det (\mv{I}_{M} + \frac{1}{\sigma_c^2}  \mv{\Sigma}^2 \tilde{\bm{Q}}_2^* )
	\le \log_2 \det (\mv{I}_{M} + \frac{1}{\sigma_c^2}  \mv{\Sigma}^2 \tilde{\bm{Q}}^{**}_2 )$ according to the Hadamard inequality \cite{horn2012matrix}, $\operatorname{tr} \{(\tilde{\bm{Q}}_2^{**})^{-1}\} \leq \operatorname{tr} \{ (\tilde{\bm{Q}}_2^*)^{-1}\} \leq \tilde{\Gamma}_2$ based on \cite[Lemma 1]{ohno2004capacity}, and $\operatorname{tr} (\tilde{\bm{Q}}_2^{**}) = \operatorname{tr} (\tilde{\bm{Q}}_2^*) \leq P$. 
	From these three inequalities, we can infer that $\tilde{\bm{Q}}_2^{**}$ is also feasible for (P2.2) and achieves a no lower objective value than that by $\tilde{\bm{Q}}_2^{*}$. This contradicts the presumption that the non-diagonal matrix $\tilde{\bm{Q}}_2^*$ is optimal. This thus verifies that the optimal solution to (P2.2) should be diagonal, i.e., $\tilde{\bm{Q}}_2 = \operatorname{diag}(p_{2,1},...,p_{2,M})$. Together with the fact that $\tilde{\bm{Q}}_2 \succ \bm{0}$, we have $p_{2,k} > 0, \forall k \in \{1,\ldots,M\}$. This thus completes the proof.
	
	\subsection{Proof of Proposition \ref{pro:prime_dual_relationship}}\label{Proof:prime_dual_relationship}
	
	We prove this proposition via the Lagrange duality method. Let $\mu_2 \geq 0$ and $v_2 \geq 0$ denote the dual variables associated with the CRB constraint and the power constraint in (\ref{equ:P_1P}), respectively. By denoting $\bm{p}_2 \triangleq \left[p_{2,1},...,p_{2,M}\right]^T$, the partial Lagrangian of (P2.3) is expressed as
	\begin{align}
		\nonumber
		\mathcal{L}_2(\bm{p}_2,\mu_2,v_2)  = \sum_{k=1}^r \log_2 \left(1+\frac{\zeta_k^2(\bm{H}_c) p_{2,k}}{\sigma_c^2} \right) - \mu_2 (\sum_{k=1}^M \frac{1}{p_{2,k}} - \tilde{\Gamma}_2) - v_2(\sum_{k=1}^M p_{2,k} - P),
	\end{align}
	and the corresponding dual function is given by
	\begin{align}\label{equ:dual_function}
		g_2(\mu_2,v_2) = \max_{\bm{p}_2 \geq \bm{0}} \mathcal{L}_2(\bm{p}_2,\mu_2,v_2).
	\end{align}
	Accordingly, the dual problem of (P2.3) is given by
	\begin{align}\label{P_1_dual}
		\text{(D2.3)}: \min _{\mu_2 \geq 0, v_2 \geq 0} g_2(\mu_2,v_2).
	\end{align}
	Since problem (P2.3) is convex and satisfies the Slater's condition, the strong duality holds between problem (P2.3) and its dual problem (D2.3) \cite{boyd2004convex}. As a result, we can optimally solve problem (P2.3) by equivalently solving the dual problem (D2.3). In the following, we first solve  (\ref{equ:dual_function}) to obtain the dual function $g_2(\mu_2,v_2)$ and then solve (D2.3) to obtain $\mu_2^{\text{opt}}$ and $v_2^{\text{opt}}$. 
	
	First, consider problem (\ref{equ:dual_function}) with given $\mu_2 \geq 0$ and $v_2 \geq 0$, and suppose that its optimal solution is given by $\bm{p}_2^*$.  By setting the partial derivatives of $\mathcal{L}_2(\bm{p}_2,\mu_2,v_2)$ with respect to $p_{2,k}$'s to be zero, we have
	\begin{align}\label{equ:Lag_zero1}
		\frac{\partial \mathcal{L}_2}{\partial p_{2,k}^*} & = \frac{1}{\text{ln}2} (\frac{1}{1+\frac{\zeta_k^2(\bm{H}_c) p_{2,k}^*}{\sigma_c^2}}) \frac{\zeta_k^2(\bm{H}_c)}{\sigma_c^2} + (\frac{\mu_2}{(p_{2,k}^*)^2}) - v_2 = 0, \quad
		\forall k \in \{1,..., r\}, \\
		\label{equ:Lag_zero2}
		\frac{\partial \mathcal{L}_2}{\partial p_{2,k}^*} & = \frac{\mu_2}{(p_{2,k}^{*})^2} - v_2 = 0, \quad \forall k \in \{r+1,..., M\}.
	\end{align}
	Based on (\ref{equ:Lag_zero1}), (\ref{equ:Lag_zero2}), and Cardano's formula for solving a cubic equation, we have the optimal solution to problem (\ref{equ:dual_function}) as
	\begin{align}
		\label{equ:general_exp_normal}
		p_{2,k}^{*} & = -t_{1,k} + \sqrt[3]{-t_{2,k}+\sqrt{t_{2,k}^2+t_{3,k}^3}} +  \sqrt[3]{-t_{2,k}-\sqrt{t_{2,k}^2+t_{3,k}^3}}, \quad \forall k \in \{1,\ldots, r\},  \\
		\label{equ:vanish_exp_normal}
		p_{2,k}^{*} & = \sqrt{\mu_2/v_2}, \quad \forall  k \in \{r+1,\ldots, M\}, 
	\end{align}
	where
	$
	t_{1,k} = b_k/(3a), \text{ } t_{2,k} = (27a^2d_k-9ab_kc+2b_k^3)/(54a^3), \text{ } t_{3,k} = (3ac-b_k^2)(9a^2),
	$
	with $a = v_2, b_k = v_2 \frac{\sigma_c^2}{\zeta_k^2(\bm{H}_c)} - \frac{1}{\text{ln2}}, c = -\mu_2$, and $d_k = -\mu_2 \frac{\sigma_c^2}{\zeta_k^2(\bm{H}_c)}$.
	
	Next, we solve the dual problem (D2.3) to find the optimal dual solution $(\mu_2^{\text{opt}},v_2^{\text{opt}})$. similar to that in Section \ref{section:point_target}, with the subgradient of $g_2(\mu_2,v_2)$ given as $\partial g_2 |_{(\mu_2,v_2)} = [ -(\sum_{k=1}^{M} \frac{1}{p_{2,k}^*} - \tilde{\Gamma}_2 ), -(\sum_{k=1}^{M} p_{2,k}^* - P)]^T$, we can apply ellipsoid method to obtain the optimal dual solution $(\mu_2^{\text{opt}},v_2^{\text{opt}})$ to (D2.3).
	Finally, by substituting $(\mu_2^{\text{opt}},v_2^{\text{opt}})$ into in (\ref{equ:general_exp_normal}) and (\ref{equ:vanish_exp_normal}), the optimal solution to (P2.3) is given in (\ref{equ:P_2_3_power}). This thus completes the proof.
	
	\subsection{Proof of Proposition \ref{Pro:diagonal_optimal_P3}}\label{Proof:prop_Eig_diag_optimal}
	Suppose that $\tilde{\bm{Q}}^*_3$ is the optimal solution to (P3.1) with non-zero non-diagonal elements, and we set an alternative solution as $\tilde{\bm{Q}}^{**}_3 =  \tilde{\bm{Q}}^*_3 \circ \bm{I}_M$. Then, according to  Hadamard's inequality, $\tilde{\bm{Q}}^{**}_3$ will further increase the objective function value, while meeting the power constraint.
	Furthermore, let $\mv{d} = [d_1,...,d_n]^T$ and $\mv{\lambda} = [\lambda_1,...,\lambda_n]^T$ be the diagonal elements of $\tilde{\bm{Q}}^*_3$ and the eigenvalues of $\tilde{\bm{Q}}^*_3$, respectively. Since $\tilde{\bm{Q}}^*_3$ is Hermitian, according to the Schur-Horn Theorem \cite{horn1954doubly}, we have
	\begin{align}
		\sum_{q=1}^{k} d_{[q]} \leq \sum_{q=1}^{k} \lambda_{[q]}, \forall k \in \{1,2,...,n\},
	\end{align}
	where $d_{[1]},...,d_{[n]}$ denotes the non-increasing rearrangement of all the elements in $\bm{d}$, i.e., $d_{[1]} \geq ... \geq d_{[n]}$ and the equality holds for $k=n$, which implies that $ d_{[n]} \geq \lambda_{[n]}$. Thus, the minimum eigenvalue of $\tilde{\bm{Q}}^{**}_3$ is equal to $d_{[n]}$, which is larger than the minimum eigenvalue of $\tilde{\bm{Q}}^{*}_3$. As a result, $\tilde{\bm{Q}}^{**}_3$ also satisfies the  CRB constraint in (P3.1). The above yields a contradiction to the presumption that $\tilde{\bm{Q}}^*_3$ is the optimal solution to (P3.1). Therefore, the optimal solution to (P3.1) must be a diagonal matrix.

	\subsection{Proof of Proposition \ref{pro:prime_dual_relationship_eig}}\label{Proof:prop_Eig_semi_form}
	Let $\{\mu_{3,k}^{\text{opt}}\}$ and $v_3^{\text{opt}}$ be the optimal dual varaibles associated with the CRB constraint and the power constraint in (P3.2), respectively. As problem (P3.2) is convex and satisfies the Slater's condition, the strong duality holds between (P3.2) and its Lagrange dual problem. The corresponding Lagrangian is
	\begin{align}
		\nonumber
		\mathcal{L}_3(\bm{p}_3,\{\mu_{3,k}\},v_3)  = \sum_{k=1}^r \log_2 \left(1+\frac{\zeta_k^2(\bm{H}_c) p_{3,k}}{\sigma_c^2} \right) +  \sum_{k=1}^M \mu_{3,k} (p_{3,k} - \tilde{\Gamma}_e) - v_3(\sum_{k=1}^M p_{3,k} - P).
	\end{align}
	Furthermore, according to KKT conditions, we have
	\begin{align} \label{equ:Lag_zero1_eig_opt}
		\frac{\partial \mathcal{L}_3}{\partial p_{3,k}^{\text{opt}}} & = \frac{1}{\text{ln}2} (\frac{1}{1+\frac{\zeta_k^2(\bm{H}_c) p_{3,k}^{\text{opt}}}{\sigma_c^2}}) \frac{\zeta_k^2(\bm{H}_c)}{\sigma_c^2} + \mu_{3,k}^{\text{opt}} - v_3^{\text{opt}} = 0, \quad
		\forall k \in \{1,..., r\},\\
		\label{equ:Lag_zero1_eig_2_opt}
		\frac{\partial \mathcal{L}_3}{\partial p_{3,k}^{\text{opt}}} & = \mu_{3,k}^{\text{opt}} - v_3^{\text{opt}} = 0, \quad
		\forall k \in \{r+1,..., M\}.
	\end{align}
	By the complementary slackness condition, we have $\mu_{3,k}^{\text{opt}} (p_{3,k}^{\text{opt}} - \tilde{\Gamma}_e) = 0$. Thus, for $k \in \{1,..., r\}$, if $p_{3,k}^{\text{opt}} > \tilde{\Gamma}_e$, then $\mu_{3,k}^{\text{opt}} = 0$ and we thus have $p_{3,k}^{\text{opt}} = \frac{1}{v_3^{\text{opt}} \ln 2} - \frac{\sigma_c^2}{\zeta_k^2(\bm{H}_c)}$; otherwise, $p_{3,k}^{\text{opt}} = \tilde{\Gamma}_e$. Next, from (\ref{equ:Lag_zero1_eig_2_opt}), $\mu_{3,k}^{\text{opt}} = v_3^{\text{opt}} > 0$,  we have $p_{3,k}^{\text{opt}} = \tilde{\Gamma}_e, \forall k \in \{r+1,..., M\}$.

	\subsection{Proof of Proposition \ref{lemma:Lemma_order_p}}\label{Proof:lemma_power_allocation}
	We prove this proposition only for Scenario 2 with Trace-CRB, and the other two scenarios can be similarly proved. Based on (\ref{equ:P_2_3_power}), it is evident that $p_{2,r+1}^{\text{opt}} = ... = p_{2,M}^{\text{opt}}>0$. Therefore, to verify Proposition \ref{lemma:Lemma_order_p}, we only need to prove that $p_{2,1}^{\text{opt}} \ge p_{2,2}^{\text{opt}} \ge ...\ge p_{2,r}^{\text{opt}} \ge p_{2,r+1}^{\text{opt}}$. 
	First, we prove $p_{2,r}^{\text{opt}} \geq p_{2,r+1}^{\text{opt}}$ via contradiction. If $p_{2,r}^{\text{opt}} < p_{2,r+1}^{\text{opt}} = \sqrt{\mu_2^{\text{opt}}/v_2^{\text{opt}}}$, then we have
	\begin{align}
		0 \stackrel{(a)}{\leq} \frac{1}{\text{ln}2} (\frac{1}{1+\frac{\zeta_r^2(\bm{H}_c) p_{2,r}^{\text{opt}}}{\sigma_c^2}}) \frac{\zeta_r^2(\bm{H}_c)}{\sigma_c^2} \stackrel{(b)}{=} \mu_2^{\text{opt}}(-\frac{1}{(p_{2,r}^{\text{opt}})^2}) + v_2^{\text{opt}} \stackrel{(c)}{<} 0,
	\end{align}
	where (a) follows from $p_{2,r}^{\text{opt}} > 0$, (b) is obtained based on (\ref{equ:Lag_zero1}), and (c) holds based on 
	the above presumption. This incurs a contradiction. Thus, we have $p_{2,r}^{\text{opt}} \geq p_{2,r+1}^{\text{opt}}$.
	Next, we prove $p_{2,k}^{\text{opt}} \ge p_{2,k+1}^{\text{opt}}, \forall k \in \{1,..., r-1\}$, by contradiction. If $p_{2,k}^{\text{opt}} < p_{2,k+1}^{\text{opt}}$, then based on (\ref{equ:Lag_zero1}), we have
	\begin{align}\label{equ:assump_1}
		\frac{1}{\text{ln}2} (\frac{1}{\frac{\sigma_c^2}{\zeta_k^2(\bm{H}_c)}+ p_{2,k}^{\text{opt}}})  = v_2^{\text{opt}} -  \frac{\mu_2^{\text{opt}}}{(p_{2,k}^{\text{opt}})^2} < v_2^{\text{opt}} -  \frac{\mu_2^{\text{opt}}}{(p_{2,k+1}^{\text{opt}})^2} = \frac{1}{\text{ln}2} (\frac{1}{\frac{\sigma_c^2}{\zeta_{k+1}^2(\bm{H}_c)}+ p_{2,k+1}^{\text{opt}}}).
	\end{align}
	Furthermore, based on the presumption $\frac{p_{2,k}^{\text{opt}}}{\sigma_c^2} < \frac{p_{2,k+1}^{\text{opt}}}{\sigma_c^2}$ and the fact that $\frac{1}{\zeta_{k}^2(\bm{H}_c)} \leq \frac{1}{\zeta_{k+1}^2(\bm{H}_c)}$, we have
	\begin{align}\label{eqn:app:2}
		\frac{1}{\zeta_k^2(\bm{H}_c)}+\frac{p_{2,k}^{\text{opt}}}{\sigma_c^2} < \frac{1}{\zeta_{k+1}^2(\bm{H}_c)}+\frac{p_{2,k+1}^{\text{opt}}}{\sigma_c^2} 
		\Leftrightarrow 1/(\frac{1}{\zeta_k^2(\bm{H}_c)}+\frac{p_{2,k}^{\text{opt}}}{\sigma_c^2}) > 1/(\frac{1}{\zeta_{k+1}^2(\bm{H}_c)}+\frac{p_{2,k+1}^{\text{opt}}}{\sigma_c^2}).
	\end{align}
	%
	Clearly, \eqref{eqn:app:2} contradicts (\ref{equ:assump_1}), yielding $p_{2,k}^{\text{opt}} \geq p_{2,k+1}^{\text{opt}}$. Combining the above finishes the proof.
	
	\subsection{Proof of Proposition \ref{pro:P_infinite}} \label{Proof:Pro_P_infinite}
	
	First, we consider that $r = M$. In this case, when $P \rightarrow \infty$, the optimal water-filling power allocation that maximizes the sum rate in \eqref{equ:P_1P} subject to the sum power constraint in \eqref{equ:P_1P} reduces to the equal power allocation $p_{2,k} =P/M, \forall k \in \{1,\ldots,M\}$. Such power allocation is shown to minimize the estimation CRB $\sum_{k=1}^M \frac{1}{p_{2,k}}$ in constraint \eqref{equ:P_1P}. Therefore, it follows that $p_{2,k}^{\text{opt}} =P/M, \forall k \in \{1,\ldots,M\}$. 
	
	Next, we consider that $r<M$. Based on Propositions \ref{pro:prime_dual_relationship} and \ref{lemma:Lemma_order_p}, we have $p_{2,r+1}^{\text{opt}} = \ldots = p_{2,M}^{\text{opt}}$ in this case. Therefore, without loss of optimality, we use $p_s = p_{2,r+1} = \ldots = p_{2,M}$ to denote the transmit power allocated to sensing subchannels. Accordingly, problem (P2.3) is equivalently reformulated as
	\begin{subequations}\label{equ:P_12_eq}
		\begin{align}
			\text{ }  \max_{\{p_{2,k} \ge 0\}_{k=1}^r, p_s \ge 0} & \sum_{k=1}^r \log_2 \left(1+\frac{\zeta_k^2(\bm{H}_c) p_{2,k}}{\sigma_c^2} \right) \\
			\label{equ:P_1P_CRB_eq} 
			\text { s.t. }  & \sum_{k=1}^r \frac{1}{p_{2,k}} + \frac{M-r}{p_s} \leq  \tilde{\Gamma}_2 \\
			\label{equ:P_1P_Power_eq}
			& \sum_{k=1}^r p_{2,k} + (M-r) p_s \leq P.
		\end{align}
	\end{subequations}
	It follows from \eqref{equ:P_1P_CRB_eq} that $p_s \ge \frac{M-r}{\tilde{\Gamma}_2}.$
	By setting $p_s =\frac{M-r}{\tilde{\Gamma}_2}$ 
	and dropping the CRB constraint (\ref{equ:P_1P_CRB_eq}), problem \eqref{equ:P_12_eq} is reduced to the following rate maximization problem:  
	\begin{align}\label{equ:P_12_reduced}
		\text{ } \max _{\{p_{2,k} \ge 0\}_{k=1}^r}  \sum_{k=1}^r \log_2 \left(1+\frac{\zeta_k^2(\bm{H}_c) p_{2,k}}{\sigma_c^2}\right),  
		\quad
		\text { s.t. }  \sum_{k=1}^r p_{2,k} \leq P - \frac{(M-r)^2}{\tilde{\Gamma}_2},
	\end{align}
	for which the optimal value serves as an upper bound of that by \eqref{equ:P_12_eq}. As $P - \frac{(M-r)^2}{\tilde{\Gamma}_2} \rightarrow \infty$, it is clear that the equal power allocation, given by $
	p_{2,k} = \frac{1}{r} (P-\frac{(M-r)^2}{\tilde{\Gamma}_2}), \forall k\in\{1,\ldots, r\},$
	is optimal for problem \eqref{equ:P_12_reduced}. With $P\rightarrow \infty$, it can be shown that $p_{2,k} = \frac{1}{r} (P-\frac{(M-r)^2}{\tilde{\Gamma}_2}), \forall k\in\{1,\ldots, r\}$, and $p_{2,k} = \frac{M-r}{\tilde{\Gamma}_2}, \forall k \in \{r+1, \ldots, M\}$, is feasible for problem (\ref{equ:P_12_eq}) and achieves the same value as the optimal value of problem \eqref{equ:P_12_reduced}. As a result, such power allocation is optimal for (\ref{equ:P_12_eq}) and thus (P2.3). This thus completes the proof.

	\bibliographystyle{ieeetran}
	
	\bibliography{refsv2}


	

\end{document}